%% file: circlepacking.tex
\documentclass[english]{article}
\usepackage[utf8]{inputenc} % encoding of latex files
\usepackage[english]{babel} % hyphenation
\usepackage[noend]{algpseudocode}
\usepackage[hyphens]{url}
\usepackage[pdfborder={0 0 0}]{hyperref}
\usepackage{cite}
\usepackage{graphicx} 
\usepackage{caption}
\usepackage{subcaption}
\usepackage{amsmath,amssymb}
\usepackage{float}
\usepackage{todonotes}
\usepackage[slant]{complexity}
\usepackage{thm-restate}
\usepackage{xspace}
\usepackage{microtype}
\usepackage{pdfpages}
\usepackage{color}
\usepackage{amsthm}
\usepackage{authblk}
\usepackage{fullpage}
\usepackage{microtype}
\usepackage{todonotes}
\usepackage{xspace}
\usepackage{mathtools}
\usepackage{breqn}
\usepackage{thm-restate}
\usepackage{refcount}
\usepackage{tikz,pgfplots}
\pgfplotsset{compat=1.12, compat/show suggested version=false}

\newtheorem{theorem}{Definition}
\newtheorem{definition}[theorem]{Definition}
\newtheorem{lemma}[theorem]{Lemma}
\newtheorem{corollary}[theorem]{Corollary}
\newtheorem{observation}[theorem]{Observation}

\newcommand{\newtext}[1]{\color{black}#1\color{black}\xspace}
\newcommand{\new}[1]{\color{black}#1\color{black}\xspace}

\newcounter{reftmpcounter}
\newcounter{resttmpcounter}
\newcommand{\restatethm}[2]{
    \setcounterref{reftmpcounter}{#2}
    \setcounter{resttmpcounter}{\thetheorem}
    \setcounter{theorem}{\thereftmpcounter}
    \addtocounter{theorem}{-1}
    #1
    \setcounter{theorem}{\theresttmpcounter}
}

\title{Packing Disks into Disks with Optimal Worst-Case Density} 

\author[1]{S\'{a}ndor P. Fekete}
\author[1]{Phillip Keldenich}
\author[1]{Christian Scheffer}
\affil[1]{Department of Computer Science, TU Braunschweig, Germany.
	\tt$\{$s.fekete, p.keldenich, c.scheffer$\}$@tu-bs.de}

\DeclarePairedDelimiter\abs{\lvert}{\rvert}% absolute value signs/area

\begin{document}

\maketitle

\begin{abstract}
We provide a tight result for a fundamental problem arising from packing  disks
into a circular container: The critical density of packing disks in a disk is
0.5. This implies that  any set of (not necessarily equal) disks of total area
$\delta\leq 1/2$ can always be packed into a disk of area 1; on the other hand, 
for any $\varepsilon>0$ there are
sets of disks of area $1/2+\varepsilon$ that cannot be packed. The proof uses a careful
manual analysis, complemented by a minor automatic part that is based on interval arithmetic.
Beyond the basic mathematical importance, our result is also useful as a
blackbox lemma for the analysis of recursive packing algorithms.
\end{abstract}

\section{Introduction}
\input{01-introduction.tex}

\section{Preliminaries}
\input{02-preliminaries.tex}

\section{A Worst-Case Optimal Algorithm}
\label{sec:algorithm}
\input{03-algorithm.tex}

\section{Analysis of the Algorithm}
\label{sec:analysis}
\input{04-analysis.tex}

\section{Hardness}
\label{sec:hardness}
\input{05-hardness.tex}

\input{04A-analysis.tex}

\section{Conclusions}
\label{sec:conc}
\input{06-conclusions.tex}

\bibliography{references}
\end{document}

%% file: 01-introduction.tex
Deciding whether a set of disks can be packed into a given container is a fundamental 
geometric optimization problem that has attracted considerable attention.
Disk packing also has numerous applications in engineering, science, operational research and everyday life, 
e.g., for the design of digital modulation schemes \cite{PWMD1992packing},
packaging cylinders \cite{CKP2008solving,fraser1994integrated},
bundling tubes or cables \cite{WHZX2002improved,SSSKK2004disk},
the cutting industry \cite{SMCSCG2007new}, or
the layout of control panels \cite{CKP2008solving},
or radio tower placement \cite{SMCSCG2007new}.
Further applications stem from chemistry \cite{WMP1994history},
foresting \cite{SMCSCG2007new},
and origami design \cite{lang1996computational}.

Like many other packing problems, disk packing is typically quite difficult;
what is more, the combinatorial hardness is compounded by
the geometric complications of dealing with irrational coordinates that arise when
packing circular objects. This is reflected by the limitations of provably optimal
results for the optimal value for the smallest sufficient disk container 
(and hence, the densest such disk packing in a disk container), a problem that was
discussed by Kraviz~\cite{kraviz67} in 1967:
Even when the input consists of just 13 unit disks, the optimal value for the densest
disk-in-disk packing was only established in 2003~\cite{13disks},
while the optimal value for 14 unit disks is still unproven. The enormous
challenges of establishing densest disk packings are also illustrated by a long-standing 
open conjecture by Erd\H{o}s and Oler from 1961~\cite{oler} regarding optimal packings of $n$ 
unit disks into an equilateral triangle, which has only been proven up to $n=15$. 
For other examples of mathematical work on densely packing relatively small numbers
of identical disks, see~\cite{goldberg71,melissen94,19disks,12disks}, and
\cite{reis75,lubachevsky97,graham98} for related experimental work.
Many authors have considered heuristics for circle packing problems, see
\cite{SMCSCG2007new,HM2009literature} for overviews of numerous heuristics and
optimization methods.  The best known solutions for packing equal disks into
squares, triangles and other shapes are continuously published on Specht's
website \url{http://packomania.com} \cite{specht2015packomania}.
		
For the case of packing not necessarily equal disks into a square container, 
Demaine, Fekete, and Lang in 2010~\cite{DFL2010circle}
showed that deciding whether a given set of disks can be packed is 
$\mathsf{NP}$-hard by using a reduction from \textsc{3-Partition}.
This means that there is (probably) no deterministic polynomial-time algorithm
that can decide whether a given set of disks can be packed into a given
container.  

On the other hand, the literature on exact approximation algorithms which
actually give performance guarantees is small.  Miyazawa et
al.~\cite{MPSSW2014polynomial} devised asymptotic polynomial-time approximation
schemes for packing disks into the smallest number of unit square bins.  More
recently, Hokama, Miyazawa, and Schouery~\cite{HMS2016bounded} developed a
bounded-space competitive algorithm for the online version of that problem.

The related problem of packing square objects has also been studied for a long
time.  The decision problem whether it is possible to pack a given set of
squares into the unit square was shown to be strongly $\mathsf{NP}$-complete by
Leung et al.~\cite{LTWYC1990packing}, also using a reduction from
\textsc{3-Partition}. Already in 1967, Moon and Moser~\cite{MM1967some} found
a sufficient condition. They proved that it is possible to pack a set of
squares into the unit square in a shelf-like manner if their combined area, the
sum of all squares' areas, does not exceed $\frac{1}{2}$. 
At the same time, $\frac{1}{2}$ is the \emph{largest upper area bound} one can hope for, 
because two squares larger than the quarter-squares shown in \new{Figure}~\ref{fig:square} 
cannot be packed. We call the ratio between the largest combined object area that can always be 
packed and the area of the container the problem's \emph{critical density}, or \emph{\new{optimal} worst-case density}.

\begin{figure}
        \centering
	\includegraphics[width=14cm]{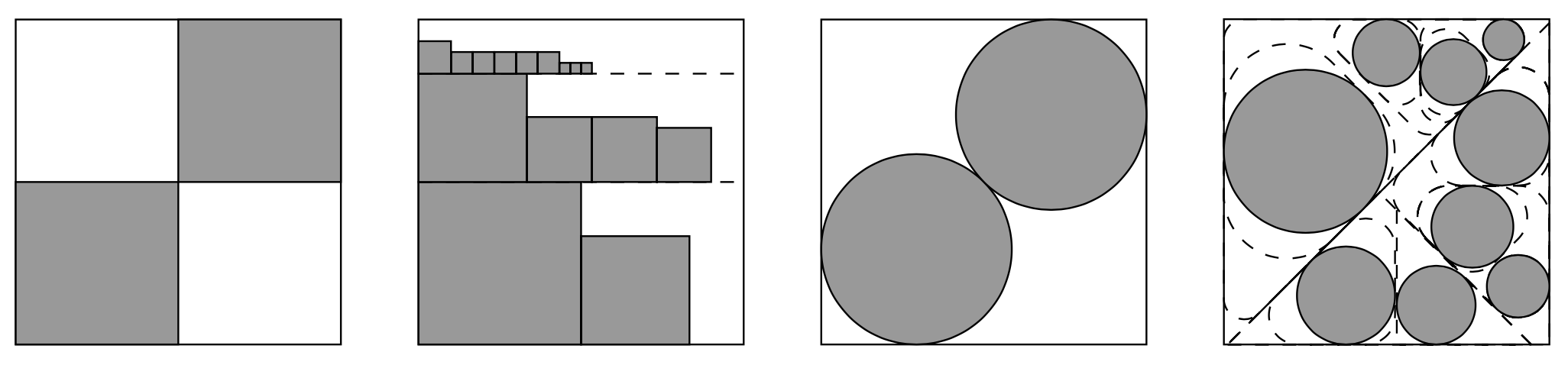}
        \caption{(1) An instance of critical density for packing squares into a square. 
		(2) An example packing produced by Moon and Moser's shelf-packing.
        	(3) An instance of critical density for packing disks into a square.
        	(4) An example packing produced by Morr's Split Packing.}
        \label{fig:square}
\end{figure}

The equivalent problem of establishing the critical packing density for disks in a square was posed by
Demaine, Fekete, and Lang~\cite{DFL2010circle} and resolved by Morr, Fekete and Scheffer~\cite{morr2017split,Fekete2018}.
Making use of a recursive procedure for cutting the container into triangular pieces,
they proved that the critical packing density of disks in a square is $\frac{\pi}{3+2\sqrt{2}} \approx 0.539$.

It is quite natural to consider the analogous question of establishing the critical packing density for 
disks in a disk. However, the shelf-packing approach of Moon and Moser~\cite{MM1967some} uses the 
fact that rectangular shapes of the packed objects fit well into parallel shelves, which is not the 
case for disks; on the other hand, the split packing method of Morr et al.~\cite{morr2017split,Fekete2018} relies
on recursively splitting triangular containers, so it does not work for a circular container that cannot
be partitioned into smaller circular pieces.

\subsection{Results}

We prove that the critical density for packing disks into a disk is 1/2:
Any set of not necessarily equal disks with a combined area of not more than 
half the area of a circular container can be packed; 
this is best possibly, as for any $\varepsilon>0$ there are instances 
of total area $1/2+\varepsilon$ that cannot be packed. See Fig.~\ref{fig:worst_case} for the critical configuration.

\begin{figure}[t]
  \begin{center}
      \includegraphics[height=5cm]{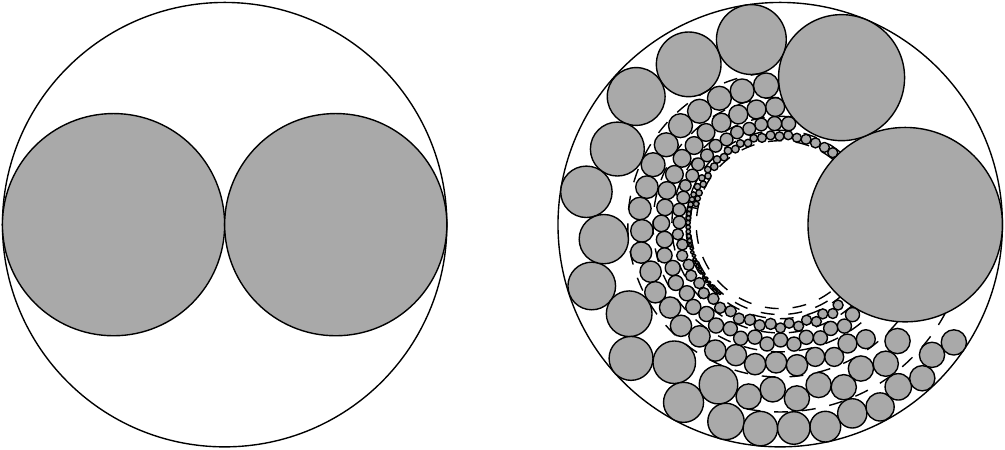} 
  \end{center}
  
  \caption{(1) A critical instance that allows a packing density no better than $\frac{1}{2}$. (2) An example packing produced by our algorithm.}
  \label{fig:worst_case}
\end{figure}

Our proofs are constructive, so they can also be used as a constant-factor
approximation algorithm for the smallest-area container of a given shape in
which a given set of disks can be packed. Due to the higher geometric difficulty
of fitting together circular objects, the involved methods are considerably more
complex than those for square containers. We make up for this difficulty
by developing more intricate recursive arguments, including appropriate
and powerful tools based on {\em interval arithmetic}.

%% file: 02-preliminaries.tex
Let $r_1,\dots,r_n$ be a set of disks in the plane. Two point sets $A,B \subset \mathbb{R}^2$ \emph{overlap} if their interiors have a point in common. A \emph{container disk} $\mathcal{C}$ is a disk that may overlap with disks from $\{ r_1,\dots,r_n \}$. The \emph{original} container disk $O$ is the unit disk. Due to recursive calls of our algorithm there may be several container disks that lie nested inside each other. Hence, the largest container disk will be the unit disk $O$. For simplification, we simultaneously denote by $r_i$ or $\mathcal{C}$ the disk with radius $r_i$ or $\mathcal{C}$ and its radius. Wl.o.g., we assume $r_1 \geq \dots \geq r_n$. We \emph{pack} the disks $r_1,\dots,r_n$ by positioning their \new{centers} inside a container disk such that $r_i$ lies inside $\mathcal{C}$ and two disks from $\{  r_1,\dots,r_n\}$ do not overlap. Given two sets $A \subseteq B \subseteq \mathbb{R}^2$, we say that $A$ is a \emph{sector} of $B$. Furthermore, we denote the volume of a point set $A$ by $\abs{A}$.

%% file: 03-algorithm.tex
\begin{theorem}\label{thm:disk_packing}
	Every set of disks with total area $\frac{\pi}{2}$ can be packed into the unit disk~$O$ with radius $1$. This induces a worst-case optimal packing density of $\frac{1}{2}$, i.e., a ratio of $\frac{1}{2}$ between the area of the unit disk and the total area to be packed.
\end{theorem}

The worst case consists of two disks $D_1,D_2$ with radius $\frac{1}{2}$, see Fig.~\ref{fig:worst_case}. The total area of these two disks is $\frac{\pi}{4} + \frac{\pi}{4}  = \frac{\pi}{2}$, while the smallest disk containing $D_1,D_2$ has an area of $\pi$.

%\section{The Algorithm}\label{sec:description}

In the remainder of Section~\ref{sec:algorithm}, we give a constructive proof for Theorem~\ref{thm:disk_packing}.
Before we proceed to describe our algorithm in Section~\ref{sec:description}, we give some definitions and describe \emph{\new{Boundary} Packing} and \emph{Ring Packing} as two subroutines of our algorithm.

\subsection{Preliminaries for the Algorithm}

We make use of the following definitions, see Fig.~\ref{fig:definition_ring}. 

\begin{figure}[h!]
  \begin{center}
      \includegraphics[scale=1]{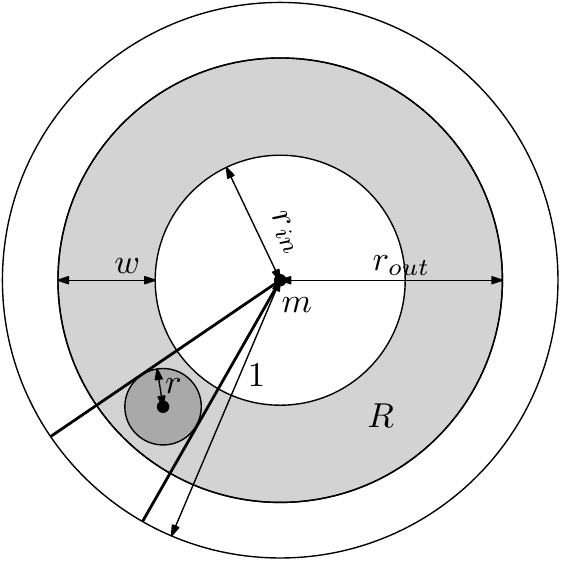} 
  \end{center}
  
  \caption{A ring $R \subset O$ with width $w$ and a disk with its corresponding tangents.}
  \label{fig:definition_ring}
\end{figure}

For \new{$r_{\text{out}} > r_{\text{in}} > 0$} and a container disk $\mathcal{C}$ such that $r_{\text{out}} \leq 2 r_{\text{in}}$, we define a {\em ring} $R:=R[r_{\text{out}}, r_{\text{in}}]$ of $\mathcal{C}$ as the closure of $r_{\text{out}} \setminus r_{\text{in}}$, see Fig.~\ref{fig:definition_ring}. The boundary of $R$ consists of two connected components. The \emph{inner boundary} is the component that lies closer to \new{the center} $m$ \new{of $r_{\text{out}}$} and the \emph{outer boundary} is the other component. The \emph{inner radius} and the \emph{outer radius} of $R$ are the radius of the inner boundary and the radius of outer boundary. Each ring \new{is associated with} one of three states $\{ \textsc{open}, \textsc{closed}, \textsc{full} \}$. Initially, each ring is \textsc{open}.

Let $r$ be a disk inside a container disk $\mathcal{C}$. The two \emph{tangents} of $r$ are the two rays starting in the midpoint of $\mathcal{C}$ and touching the boundary of $r$. We say that a disk lies \emph{adjacent} to $r_{\text{out}}$ when the disk is touching the boundary of $r_{\text{out}}$ from the inside of $r_{\text{out}}$.

%\begin{figure}[h!]
%  \begin{center}
%      \includegraphics[height=4cm]{figures/approach.pdf} 
%  \end{center}
%  \caption{TODO:new figure\textbf{Left:} If $r_1,r_2 \geq 0.495$ we place $r_1,r_2$ explicitly and recurse on another disk $r_3$ with the remaining disks. \textbf{Middle:} Placing the black disk $r_i$ touching the outer boundary of the ring $R$ from the interior of $R$. \textbf{Right:} A solution computed by our approach.}
%  \label{fig:approach}
%\end{figure}

\subsection{Boundary Packing: A Subroutine}\label{sec:diskpacking}
	
\begin{figure}[h!]
  \begin{center}
      \includegraphics[scale=1]{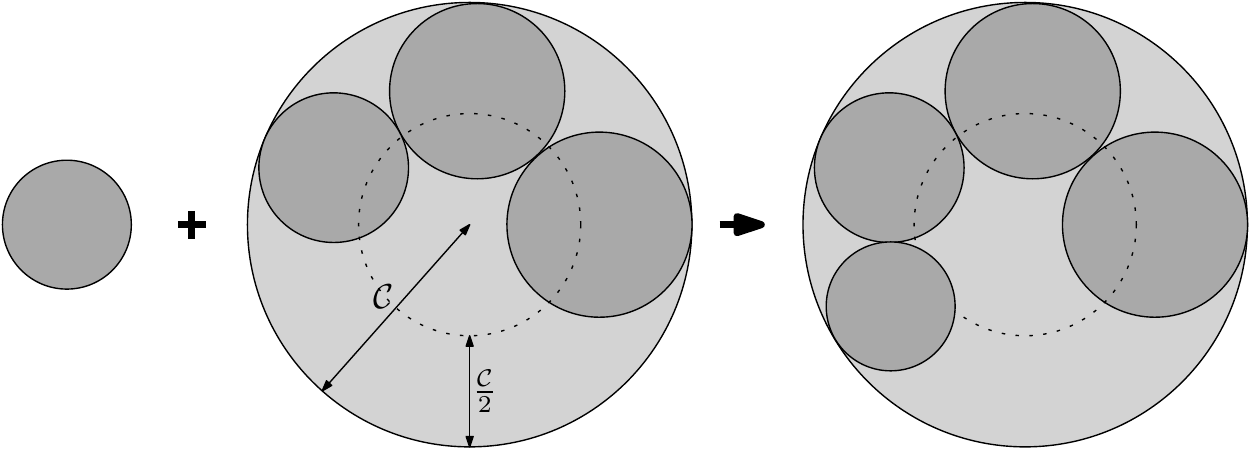} 
  \end{center}
  \caption{\new{Boundary} Packing places disks into a container disk $\mathcal{C}$ adjacent to the boundary of $\mathcal{C}$ as long as the diameter of the disks to be packed is at least as large as \new{a given threshold $\mathcal{T}$} or until \new{the current disk} does no longer fit into~$\mathcal{C}$. \new{Initially, we have $\mathcal{T} = \frac{1}{4}$.}}
  \label{fig:diskpacking}
\end{figure}
	
	Consider a container disk $\mathcal{C}$, a (possibly empty) set $S$ of already packed disks that overlap with $\mathcal{C}$, and another disk $r_i$ to be packed, see Fig.~\ref{fig:diskpacking}. We \emph{pack} $r_i$ into $\mathcal{C}$ adjacently to the boundary of $\mathcal{C}$ as follows: Let $\alpha$ be the maximal polar angle realized by a midpoint of a disk from $S$. We choose the midpoint of $r_i$ realizing the smallest possible polar angle $\beta \geq \alpha$ such that $r_i$ touches the outer boundary of $\mathcal{C}$ from the interior of $\mathcal{C}$ without overlapping another disk from $S$, see Fig.~\ref{fig:diskpacking}. If $r_i$ cannot be packed into $\mathcal{C}$, we say that $r_i$ \emph{does not fit into $R$}.
	
	Let $0 < \mathcal{T} \leq \frac{1}{4}$, called \new{the} \emph{threshold}. \emph{\new{Boundary} Packing} iteratively packs disks in decreasing order into $\mathcal{C}$ until the current disk $r_i$ does not fit into $\mathcal{C}$ or the radius of $r_{i}$ is smaller than $\mathcal{T}$.
	
	\subsection{\textsc{Ring Packing}: A Subroutine}\label{sec:ringpacking}
	
\begin{figure}[b]
  \begin{center}
      \includegraphics[scale=1]{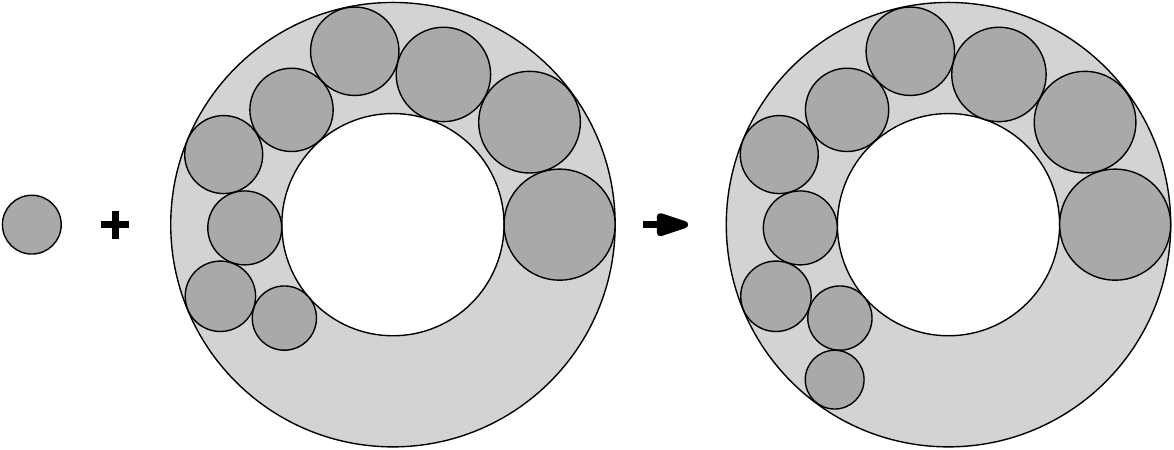} 
  \end{center}
  \caption{Ring Packing packs disks into a ring $R[r_{\text{out}},
	r_{\text{in}}]$, alternating adjacent to the outer and to the inner
	boundary of $R$.}
\label{fig:ringpacking}
\end{figure}
	
	Consider a ring $R:= R[r_{\text{out}}, r_{\text{in}}]$ with inner radius $r_{\text{in}}$ and outer radius $r_{\text{out}}$, a (possibly empty) set $S$ of already packed disks that overlap with $R$, and another disk $r_i$ to be packed, see Fig.~\ref{fig:ringpacking}. We \emph{pack} $r_i$ into $R$ adjacent to the outer (inner) boundary of $R$ as follows: Let $\alpha$ be the maximal polar angle realized by a midpoint of a disk from $S$. We choose the midpoint of $r_i$ realizing the smallest possible polar angle $\beta \geq \alpha$ such that $r_i$ touches the outer (inner) boundary of $R$ from the interior of $R$ without overlapping another disk from $S$. If $r_i$ cannot be packed into $R$, we say that $r_i$ \emph{does not fit into $R$ (adjacent to the outer (inner) boundary)}. 
	
	 \emph{Ring Packing} iteratively packs disks into $R$ alternating adjacent to the inner and outer boundary. If the current  disk $r_i$ does not fit into $R$ Ring Packing stops and we declare $R$ to be \textsc{full}. If $r_{i-1}$ and $r_{i}$ could \new{pass} each other, i.e., the sum of the diameters of $r_{i-1}$ and $r_i$ are smaller than the width of $R$, Ring Packing stops and we declare $R$ to be \textsc{closed}.
	
\subsection{Description of the Algorithm}\label{sec:description}

\begin{figure}[h!]
  \begin{center}
      \includegraphics[scale=1]{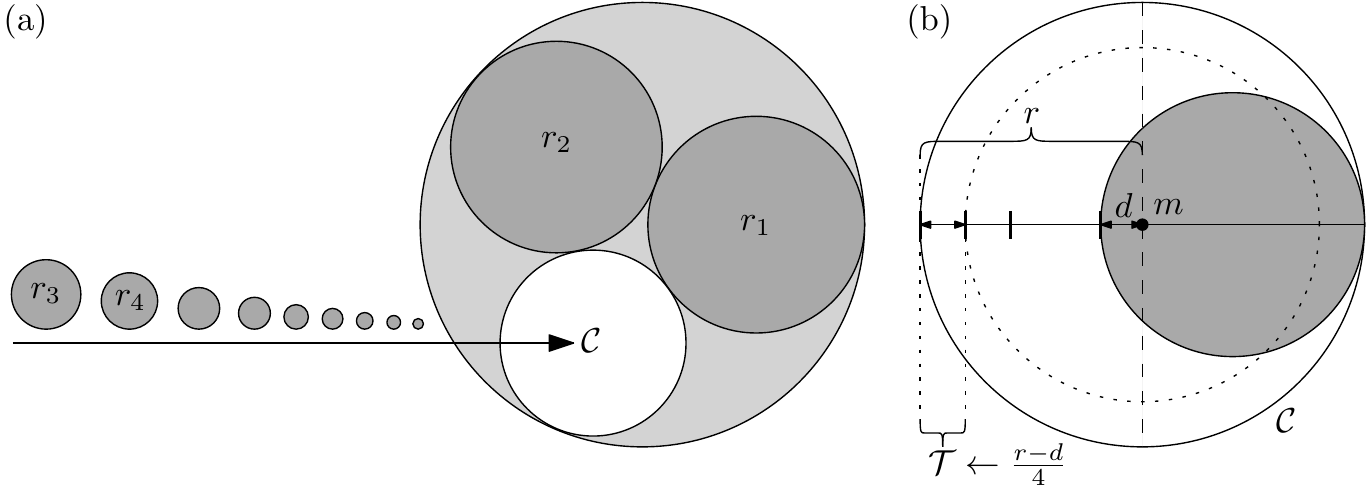} 
  \end{center}
  \caption{(a): If $r_1,r_{2} \geq 0495 \mathcal{C}$, \new{Boundary} Packing packs $r_1,r_{2}$ into $\mathcal{C}$. We update the current container disk $\mathcal{C}$ as the largest disk that fits into $\mathcal{C}$ and recurse on $\mathcal{C}$ with $r_{3}, \dots, r_{n}$. (b): Determining the threshold $\mathcal{T}$ for disks packed by \new{Boundary} Packing.}
  \label{fig:recursionapproach}
\end{figure}
	
	Our algorithm \emph{creates} rings. A ring only exists after it is created. We stop packing at any point in time when all disks are packed. Furthermore, we store the current threshold $\mathcal{T}$ for \new{Boundary} Packing and the smallest inner radius $r_{\min}$ of a ring created during the entire run of our algorithm. Initially, we set $\mathcal{T} \leftarrow \frac{1}{2}, r_{\min} \leftarrow 1$. Our algorithm works in five phases:
	\begin{itemize}
		\item \textbf{Phase 1 - Recursion:} If $r_1,r_{2} \geq 0.495 \mathcal{C}$, apply \new{Boundary} Packing to $r_{1},r_{2}$, update $\mathcal{C}$ as the largest disk that fits into $\mathcal{C}$ and $\mathcal{T}$ as the radius of $\mathcal{C}$, and recurse on $\mathcal{C}$, see Fig.~\ref{fig:recursionapproach}.
		
		\item \textbf{Phase 2 - \new{Boundary} Packing:} Let $r$ be the radius of $\mathcal{C}$. If the midpoint $m$ of $\mathcal{C}$ lies inside a packed disk $r_i$, let $d$ be the minimal distance of $m$ to the boundary of $r_i$, see Fig.~\ref{fig:recursionapproach}(b). Otherwise, we set $d = 0$.
		
		 We apply \new{Boundary} Packing to the container disk $\mathcal{C}$ with the threshold $\mathcal{T} \leftarrow \frac{r-d}{4}$.
			
		\item \textbf{Phase 3 - Ring Packing:} We apply Ring Packing to the ring $R:= R[r_{\text{out}}, r_{\text{in}}]$ determined as follows: Let $r_i$ be the largest disk not yet packed. If there is no open ring inside $\mathcal{C}$, we create a new open ring $R[r_{\text{out}}, r_{\text{in}}] \leftarrow R[r_{\min}, r_{\min} - 2r_i]$. Else, let $R[r_{\text{out}}, r_{\text{in}}]$ be the open ring with the largest inner radius $r_{\text{in}}$.
					
				\item \textbf{Phase 4 - Managing Rings:} Let $R[r_{\text{out}}, r_{\text{in}}]$ be the ring filled in Phase 3. We declare $R[r_{\text{out}}, r_{\text{in}}]$ to be closed and proceed as follows: Let $r_i$ be the largest disk not yet packed.

			If $r_i$ and $r_{i+1}$ can pass one another inside $R[r_{\text{out}}, r_{\text{in}}]$, i.e., if $2r_i + 2r_{i+1} \leq r_{\text{out}} - r_{\text{in}}$, we create two new open rings $R[r_{\text{out}}, r_{\text{out}} - 2r_{i}]$ and $R[r_{\text{out}} - 2r_{i}, r_{\text{in}}]$.
			
			\item \textbf{Phase 5 - Continue:}
			If there is an open ring, we go to Phase 3. Otherwise, we set $\mathcal{C}$ as the largest disk not covered by created rings, set $\mathcal{T}$ as the radius of $\mathcal{C}$, and go to Phase 2.
		
	\end{itemize}

%% file: 04-analysis.tex
%We start by analysing the recursion. After that, we provide the remaining analysis of the algorithm.

\subsection{Analysis of Phase 1 - The Recursion}
	
	If $r_2 \geq 0.495$, Lemma~\ref{lem:recursion} allows us to recurse on $\mathcal{C}$ as required by Phase 1.

\begin{figure}[h!]
  \begin{center}
      \includegraphics[height=5cm]{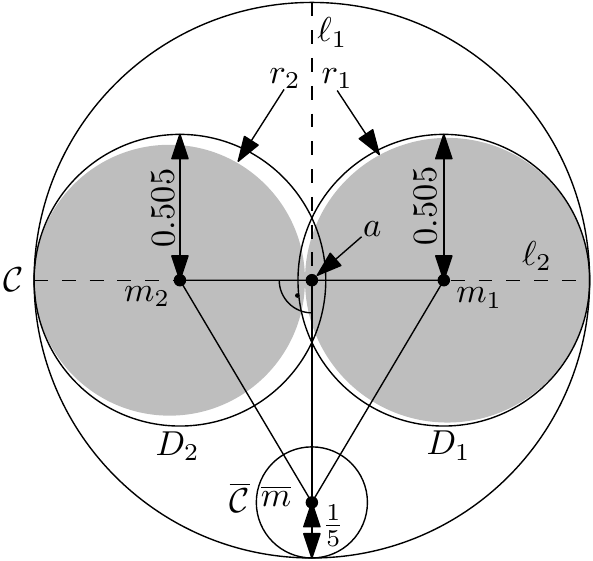} 
  \end{center}

  \caption{If $r_2 \geq 0.495$, we can pack $r_1,r_2$ into container disks $D_1,D_2$ and recurse on a third disk $\overline{c}$ whose area is twice the total area of the remaining disks.}
  \label{fig:recurse}
\end{figure}

\begin{lemma}\label{lem:recursion}
	If $r_1,r_2 \geq 0.495 \mathcal{C}$, the volume of the largest container disk that fits into $\mathcal{C}$ after packing $r_1,r_2$ is at least twice the total volume of \new{$r_{3},\dots,r_n$}, see Fig.~\ref{fig:recurse}.
\end{lemma}
\begin{proof}
	W.l.o.g., assume that the original container disk is the unit disk. Lemma~\ref{lem:twodisksonediamater} implies $r_1+r_2 \leq 1$, which
means $r_1,r_2 \leq 0.505$, because $r_2 \geq 0.495$. Furthermore, $r_1 + r_2
\leq 1$ implies that we can move (w.l.o.g.) $r_1,r_2$ into two disks $D_1,D_2$
with radius $0.505$, touching the boundary of $\mathcal{C}$ and with their
midpoints $m_1,m_2$ on the horizontal diameter of $\mathcal{C}$, see
Fig.~\ref{fig:recurse}. This decreases the volume of the largest disk that
still fits into $\mathcal{C}$. Consider the disk \new{$\overline{\mathcal{C}} :=
\frac{1}{5}$} lying adjacent to $\mathcal{C}$ and with its midpoint
$\overline{m}$ on the vertical diameter $\ell_1$ of $\mathcal{C}$. Pythagoras'
Theorem implies that $\abs{m_1\overline{m}} = \sqrt{\left( 1-0.505 \right)^2 + \left( 1
- \frac{1}{5} \right)^2} \approx 0.94075 > 0.505 + \frac{1}{5}$. Finally, we
observe that the area of $\overline{\mathcal{C}}$ is $\frac{\pi}{25} = 0.4 \pi
> 0.0199 = 2 \left(  \frac{\pi}{2} - 2 \cdot \pi 0.495^2 \right)$. This means
that the area of $\overline{\mathcal{C}}$ is twice the total area of the remaining
disks $r_3,r_4,r_5,...$, concluding the proof.  
\end{proof}

A technical key ingredient in the proof of Lemma~\ref{lem:recursion} is the following lemma:

\begin{lemma}\label{lem:twodisksonediamater}
	The area of two disks $r_1,r_2$ is at least $\frac{\pi}{2} \left(r_1+r_2 \right)^2$.
\end{lemma}
\begin{proof}
	The first derivative of the function mapping a radius onto the area of the corresponding disk is the periphery of the corresponding circle. As \new{$r_1 \geq r_2$}, decreasing \new{$r_1$} and increasing \new{$r_2$} by the same value $\delta$ reduces the total area of \new{$r_1, r_2$},
 while the value \new{$r_1+r_2$} stays the same. Hence, we assume w.l.o.g.\ that $r_1=r_2$. This implies that the total area of \new{$r_1,r_2$} is \new{$2 \pi r_1^2 = \frac{\pi}{2}\left( r_1+r_2 \right)^2$}, concluding the proof.
\end{proof}

	This allows us to assume $r_2 < 0.495 \mathcal{C}$ during the following analysis.

\subsection{Outline of the Remaining Analysis}
	
	Once our algorithm stops making recursive calls, i.e., stops applying Phase 1, Phase 1 is never applied again. W.l.o.g., let $r_1,\dots,r_n$ be the remaining disks and $O$ the container disk after the final recursion call.

	The main idea of the remaining analysis is the following: We cover the original container disk $O$ by a set of \emph{sectors} that are subsets of $O$. Let $r_i$ be a disk packed by \new{Boundary} Packing into the current container disk $\mathcal{C}$. We define the \emph{cone} induced by $r_i$ as the area of $\mathcal{C}$ between the two tangents of~$r_i$. We say that $\mathcal{C}$ is the radius of the cone. \newtext{A \emph{sector} is a subset of $O$.}
	
	 Each disk \emph{pays} portions, called \emph{atomic potentials}, of its volume to different sectors of $O$. The total atomic potential paid by a disk $r$ will be at most the volume of the disk $r$. Let $A_1,\dots,A_k$ be the total atomic potentials paid to the sectors $S_1,\dots, S_k \subset O$. The \emph{potential} of a sector $S \subseteq O$ is the sum of the proportionate atomic potentials from $S_1,\dots,S_k$, i.e., the sum of all $\frac{\abs{S_i \cap S}}{\abs{S_i}} A_i$ for $i = 1,\dots,k$. The \emph{(virtual packing) density} $\rho(S)$ of the sector $S$ is defined as the ratio between the potential of $S$ and the volume of $S$. If a sector achieves a density of $\frac{1}{2}$, we say that the sector is \emph{saturated}, otherwise its \emph{unsaturated}.
	
	Our approach for proving Theorem~\ref{thm:disk_packing} is by induction \new{over $n$}. In particular, we assume that $O \setminus \mathcal{C}$ is saturated; we show that each disk $r_i$ can be packed by our algorithm, as long as $\mathcal{C}$ is unsaturated implying that each set of disks with total volume of at most~$\frac{\abs{O}}{2}$ is packed. For a simplification, we assume for the remainder of the paper that $\mathcal{C}$ is the unit disk, i.e., $\mathcal{C} = 1$.
	
	We consider the configuration achieved after termination.

\begin{observation}\label{obs:assumeallringsfull}
	If there is a ring that is neither full nor closed, all disks are packed.
\end{observation}

\begin{figure}[t]
  \begin{center}
      \includegraphics[scale=0.93]{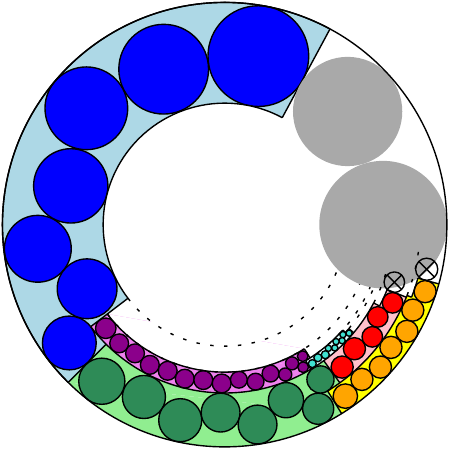} 
  \end{center}
  \caption{Different sequences of rings packed by different applications of Ring Packing. The minimal rings into which the orange and red disks are packed are full. The minimal ring into which the turquoise disks are packed is open. The uncolored, crossed-out circles illustrate that the corresponding disk did not fit into the current ring, causing it to be declared full.}
  \label{fig:statesofrings}
\end{figure}
	
	\begin{figure}[b]
  \begin{center}
      \includegraphics[scale=1]{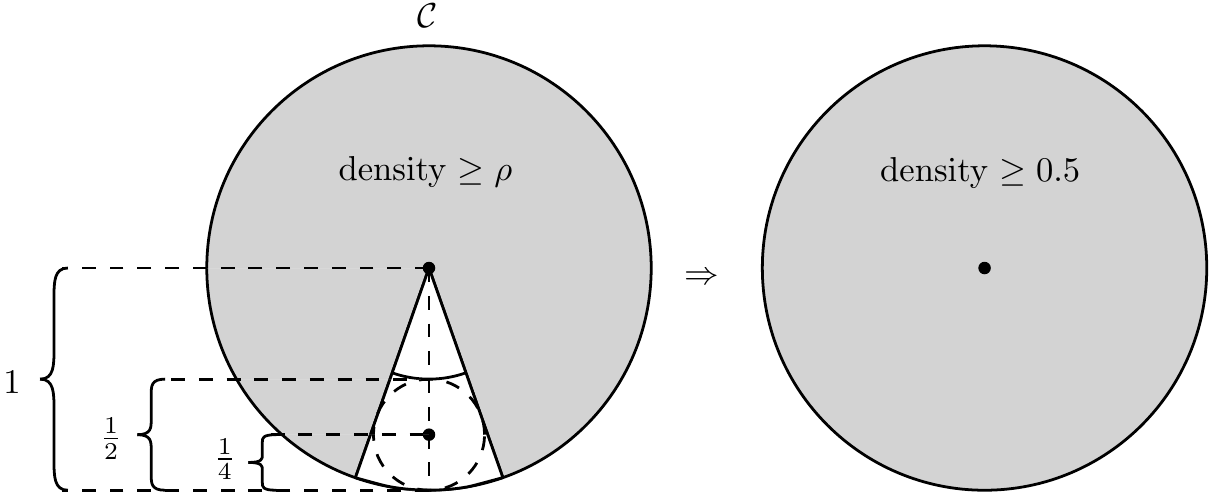} 
  \end{center}
  \caption{Ensuring a density of at least $0.5$ for a ring $R$ needs a density of $0.5606$ for $R \setminus C$.}
  \label{fig:criticaldensity}
\end{figure}

{Thus, we assume that all rings computed by our algorithm are full or closed. In order to avoid that \new{Boundary} Packing stops due to a disk $r$ not fitting, we consider the gap that is left by \new{Boundary} Packing, see
Fig.~\ref{fig:criticaldensity}. \new{This gap achieves its maximum for} $r = \frac{1}{4}$. In order to ensure that even in this case
the entire container $\mathcal{C}$ is saturated, we guarantee that $\mathcal{C}$ has a density of}

\begin{equation*}
	\rho := \frac{180^{\circ}}{360^{\circ} - 2 \arcsin \left( \frac{1/4}{3/4}
\right)} < 0.56065.
\end{equation*}

\subsection{Analysis of \new{Boundary} Packing}\label{sec:diskpacking_ana}

The following lemma is the key ingredient for the analysis of \new{Boundary} Packing.

\begin{lemma}\label{lem:densityCones}
	Let $r \in [0.2019 , \frac{1}{2}]$ be a disk lying adjacent to $\mathcal{C}$. The cone $C$ induced by $r$ has a density better than $\rho$ if $r \in [\frac{1}{4}, 0.495]$ and \new{at} least $\frac{1}{2}$ if $r \in [0.2019, \frac{1}{2}]$, see Fig.\new{~\ref{fig:r1large_main}}.
\end{lemma}
\begin{figure}[h!]
  \begin{center}
      \includegraphics[scale=1]{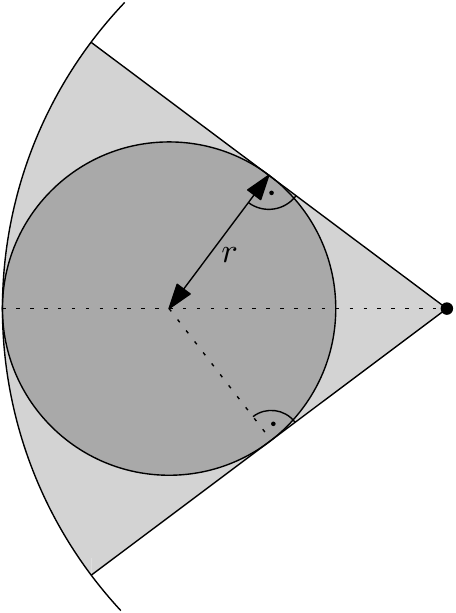} 
  \end{center}

  \caption{A disk $r \in [\frac{1}{4}, 0.495]$ lying adjacent to $\mathcal{C}$ induces a cone with density of at least $0.56127$ if $r \in [\frac{1}{4}, 0.495]$ and of least $\frac{1}{2}$ if $r \in [\frac{1}{4}, \frac{1}{2}]$.}
  \label{fig:r1large_main}
\end{figure}

\begin{proof}
	Let $f(r) := \frac{\pi r^2}{\arcsin \left( \frac{r}{1-r} \right)}$ for $\frac{1}{4} \leq r \leq \frac{1}{2}$. Thus we have 
	\begin{equation*}
		f'(r) = \frac{2 \pi r}{\arcsin \left( \frac{r}{1-r} \right)} - \frac{\pi r^2 \left( \frac{1}{1-r} + \frac{r}{\left( 1-r \right)^2} \right)}{\arcsin \left( \frac{r}{1-r} \right)^2 \sqrt{1 - \frac{r^2}{\left( 1-r \right)^2}}}.
	\end{equation*}
	Solving $f'(r) = 0$ yields $r \approx 0.39464$. Furthermore, we have $f(\frac{1}{4}) \approx 0.57776$, $f(0.39464) = 0.68902$, $f(\frac{1}{2}) = 0.5$, and $f(0.495) \approx 0.56127$. Thus, $f$ restricted to $[\frac{1}{4},0.495]$ achieves at $0.495$ its global minimum $0.56127$.
	
	A similar approach implies that $f$ restricted to $[0.2019,\frac{1}{2}]$ attains its global minimum $\frac{1}{2}$ at $\frac{1}{2}$.
\end{proof}

The following lemma proves that all disks $r_i \geq \frac{\mathcal{C}}{4}$ that are in line to be packed into a container disk $\mathcal{C}$ can indeed be packed into $\mathcal{C}$.

\begin{lemma}\label{lem:largestdisknotadjacenttoboundary}
	%If $r_1 \leq \frac{1}{2}$
	%Let $C \subset \mathcal{C}$ be an empty cone with midpoint $m$ and radius $1$, such that $\mathcal{C} \setminus C$ is saturated.
	All disks $r_i \geq \frac{1}{4}$ that are in line to be packed into~$\mathcal{C}$ by \new{Boundary} Packing do fit into~$\mathcal{C}$.
\end{lemma}
\begin{proof}
	Assume that there is a largest disk $r_k \geq \frac{1}{4}$ not packed
adjacent to $\mathcal{C}$. Each disk $r_i$ from $r_1,\dots,r_{k-1}$ pays its
entire volume to the cone induced by~$r_i$. Lemma~\ref{lem:densityCones}
implies that each cone \new{is saturated}. As $r_k$ does
not fit between $r_1,r_{k-1}$ and is adjacent to $\mathcal{C}$,
Lemma~\ref{lem:densityCones} implies that the area of $\mathcal{C}$ that is
not covered by a cone induced by $r_1,\dots,r_{k-1}$ has a volume smaller than
twice the volume of $r_k$. This implies that the total volume of $r_1,\dots,r_k,$ is
larger than half of the volume of $C$. This implies that the total input volume
of \new{$r_1,\dots,r_n$} is larger than twice the volume of the container. This is a
contradiction, concluding the proof.  
\end{proof}

	%Lemma~\ref{lem:largestdisknotadjacenttoboundary} implies that our algorithm works correctly if the smallest disk has a radius of at least~$\frac{1}{4}$.

\begin{corollary}\label{cor:alldisksfitiflarge}
	If $r_n \geq \frac{1}{4} $, our algorithm packs all input disks.
\end{corollary}

Thus, we assume w.l.o.g. $r_n < \frac{1}{4}$, implying that our algorithm creates rings.
	
\subsection{Analysis of Ring Packing}\label{sec:ringpacking_ana}

For the following definition, see Fig.~\ref{fig:ring_packing}~(Middle).
\begin{definition}
	A \emph{zipper} \new{$Z$} is a (maximal) sequence $\langle r_k,\dots,r_{\ell} \rangle$ of disks that are packed into a ring $R$ during an application of Ring Packing. The \emph{length} of $Z$ is defined as $k - \ell + 1$. 
\end{definition}

Consider a zipper $\langle r_k,\dots,r_{\ell} \rangle$ packed into a ring $R$. For a simplified presentation, we assume in Section~\ref{sec:ringpacking_ana} that the lower tangent of $r_k$ realizes a polar angle of zero, see Fig.~\ref{fig:ring_packing}.

We refine the potential assignments of zippers as follows. Let $Z = \langle r_k,\dots,r_{\ell} \rangle$ be an arbitrary zipper and $R$ the ring into which $Z$ is packed. In order to subdivide $R$ into sectors corresponding to specific parts of the zipper, we consider for each disk $r_i$ the \emph{center ray}, which is the ray starting from~$m$ and passing the midpoint of $r_i$. Let $t_1, t_2$ be two rays starting in $m$. We say that $t_1$ lies above $t_2$ when the polar angle realized by $t_1$ is at least as large as the polar angle realized by $t_2$. $t_1$ is the \emph{minimum} (\emph{maximum}) of $t_1,t_2$ if $t_1$ does not lie above (below) $t_2$. Furthermore, the \emph{upper tangent} (\emph{lower tangent}) of a disk $r_i$ is the maximal (minimal) tangent of $r_i$.

\begin{figure}[h!]
  \begin{center}
      \includegraphics[height=4.5cm]{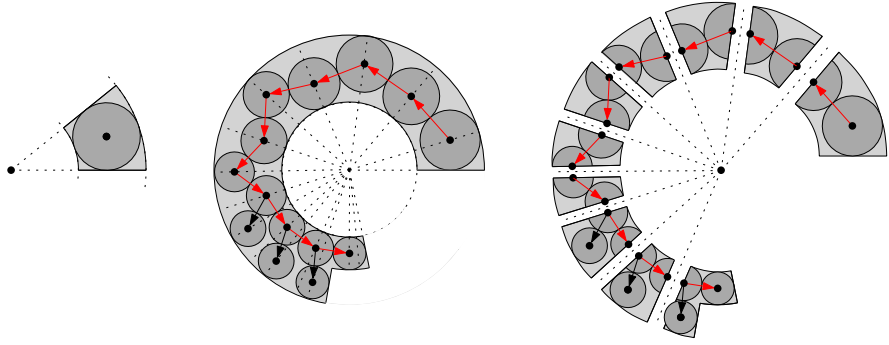} 
  \end{center}

  \caption{A maximal sequence of disks that are packed into a ring during an application of \new{Boundary} Packing. The corresponding sectors are illustrated in light gray \textbf{Left:} A zipper of \new{size} one and the corresponding sector. \textbf{Middle:} A zipper of \new{size~$14$}, the resulting directed adjacency graph (black/red), and the path (red) leading from the largest disk to the smallest disk. The first seven edges of $P$ are diagonal and the remaining edges of $P$ are vertical. \textbf{Right:} The zipper and the sector disassembled into smaller sectors corresponding to the edges of the red path.}
  \label{fig:ring_packing}
\end{figure}

If the zipper $Z$ consists of one disk $r_k$, the \emph{sector} $S$ of $Z$ is that part of $R$ between the two tangents to $r_k$ and $r_k$ pays its entire volume to $S$.

\begin{lemma}\label{lem:zipperlengthone}
	The density of the sector $S$ of a zipper of length one is at least $0.77036$.
\end{lemma}
\begin{proof}
	As the zipper consists of only one disk $r_k$, $r_k$ touches both the inner and the outer boundary of $R$. Hence, the density of $S$ is not increased by assuming that the inner radius of $R$ is equal to the diameter of $r_k$. Hence, the density of $S$ is at least $\frac{\pi}{12 \arcsin(1/3)} \approx 0.77036$.
\end{proof}

Assume the zipper $\langle r_k,\dots,r_{\ell} \rangle$ consists of at least two disks. We define the \emph{adjacency graph} $G=(\{ r_k,\dots,r_{\ell} \}, E)$ as a directed graph as follows: There is an edge $(r_j,r_i)$ if (1) $r_i \leq r_j$ and (2) $r_i, r_j$ are touching each other, see Fig.~\ref{fig:ring_packing}~(Right).  As Ring Packing packs each disk $r_i$ with midpoint $m_i$ such that $m_i$ realizes the smallest possible polar angle, there is a path $e_k,\dots,e_{\ell-1} =: P$ connecting $r_k$ to $r_{\ell}$ in the adjacency graph $G$, see Fig.~\ref{fig:ring_packing}~\new{(Middle)}. $e_k$ is the \emph{start} edge of $P$ and $e_{\ell-1}$ is the end edge of~$P$. The remaining edges of $P$ that are neither the start nor the end edge of $G$, are \emph{middle} edges of $P$. Furthermore, an edge $(r_j,r_{m}) = e_i \in P$ is \emph{diagonal} if $r_j,r_{m}$ are touching different boundary components of $R$. Otherwise, we call $e_i$ \emph{vertical}.

	Depending on whether $e_i$ is a start, middle, or an end edge and on whether $e_i$ is diagonal or vertical, we classify the edges of the path $P$ by eight different types T1-T8. For each type we individually define the sector $A_i$ belonging to an edge $(r_j,r_{m}) = e_i \in P$ and the potential assigned to $A_i$, called the \emph{potential of $e_i$}.
	
	Let $t_{\text{lower}}$ be the minimum of the lower tangents of $r_j,r_{m}$ and $t_{\text{upper}}$ the maximum of the upper tangents of $r_j,r_{m}$, see Fig.~\ref{fig:threelines} (a). Furthermore, let $t_1,t_2$ be the center rays of $r_j,r_{m}$, such that $t_1$ does not lie above $t_2$.
	
	For the case that $e_i = (r_j,r_{m})$ is a vertical edge, we consider additionally the disk $r_{p}$ that is packed into $R$ after $r_j$ and before $r_{m}$, see Fig.~\ref{fig:threelines} (f). Let $t_3$ be the maximum of $t_2$ and the upper tangent of $r_{p}$, see Fig.~\ref{fig:threelines}. 
		
\begin{figure}[h!]
  \begin{center}
      \includegraphics[scale=1]{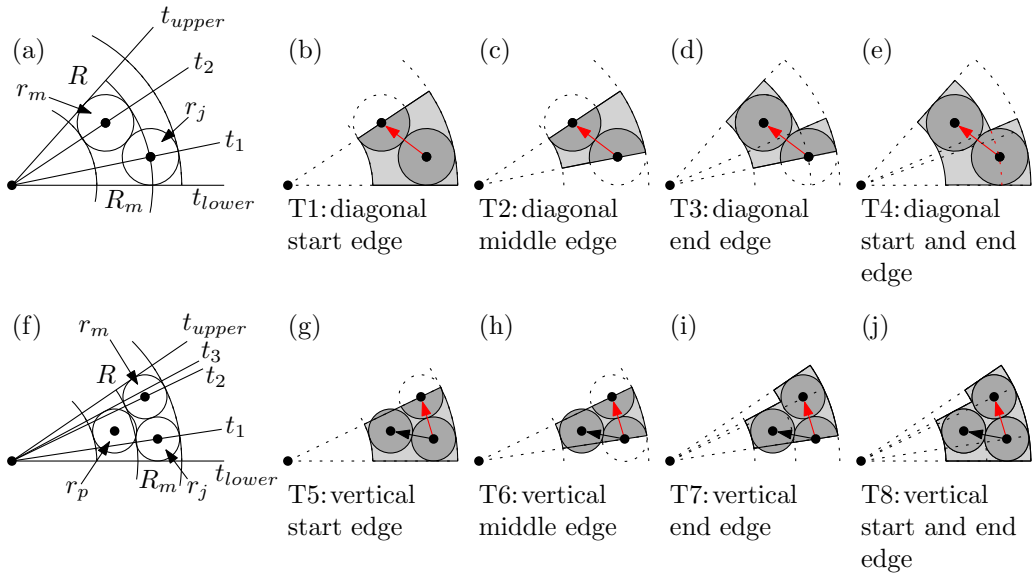} 
  \end{center}

  \caption{The eight possible configurations of an edge $e_i$ (red) of $P$, the corresponding sectors (light gray), and the potentials (dark gray) payed by the involved disks to the sector.}
  \label{fig:threelines}
\end{figure}
	\begin{itemize}
		\item[T1] \textbf{The sector of $e_i$}: If $e_i = (r_j,r_{m})$ is a diagonal start edge (as shown in Fig.~\ref{fig:threelines}(b)), the sector of $e_i$ is that part of $R$ that lies between $t_{\text{lower}}$ and $t_2$.
		
		\textbf{The potential of $e_i$}: $r_j$ pays its entire volume and $r_{m}$ the half of its volume to the sector of~$e_i$.
		
		\item[T2] \textbf{The sector of $e_i$}: If $e_i = (r_j,r_{m})$ is a diagonal middle edge,  (as shown in Fig.~\ref{fig:threelines}(c)), the sector of $e_i$ is that part of $R$ that lies between $t_1$ and $t_2$.
		
		\textbf{The potential of $e_i$}: $r_j$ and $r_m$ pay the half of its volume to the sector of $e_i$.
		
		\item[T3] \textbf{The sector of $e_i$}: If $e_i = (r_j,r_{m})$ is a diagonal end edge,  (as shown in Fig.~\ref{fig:threelines}(d)), the sector of $e_i$ consists of two parts: (1) The first is the part of $R$ that lies between the upper tangent and the center ray of~$r_j$. (2) Let $R_m$ be the smallest ring enclosing $r_m$. The second part of the sector is that part of $R_m$ that lies between the upper tangent of $r_m$ and the minimum of $t_1$ and the lower tangent of $r_m$.
		
		\textbf{The potential of $e_i$}: $r_j$ pays the half of its volume and $r_m$ its entire volume to the sector of~$e_i$.
		
		\item[T4] \textbf{The sector of $e_i$}: If $e_i =
(r_j,r_{m})$ is a diagonal start and end edge,  (as shown in
Fig.~\ref{fig:threelines}(e)), the sector of $e_i$ is
the union of two sectors: (1)~The first is the part of $R$ that lies between the lower and
the upper tangent of $r_j$. (2)~The second is that part of $R_m$ that lies between the lower
and the upper tangent of $r_m$. 
		
		\textbf{The potential of $e_i$}: $r_j,r_m$ pay their entire volume to the sector of $e_i$.
		
		\item[T5] \textbf{The sector of $e_i$}: If $e_i = (r_j,r_m)$ is a vertical start edge,  (as shown in Fig.~\ref{fig:threelines}(g)), the sector of $e_i$ is that part of $R$ that lies between the minimum of the lower tangents of $r_j,r_p$ and the center ray of $r_m$.
		
		\textbf{The potential of $e_i$}: $r_j,r_p$ pay their entire volume and $r_m$ the half of its volume to the sector of $e_i$.
		
		\item[T6] \textbf{The sector of $e_i$}: If $e_i = (r_j,r_m)$ is a vertical middle edge,  (as shown in Fig.~\ref{fig:threelines}(h)), the sector of $e_i$ is that part of $R$ that lies between the center rays of $r_j,r_m$.
		
		\textbf{The potential of $e_i$}: $r_p$ pays its entire volume and $r_j,r_m$ pay half of their respective volume to the sector of $e_i$.
		
		\item[T7] \textbf{The sector of $e_i$}: If $e_i = (r_j,r_m)$ is a vertical end edge,  (as shown in Fig.~\ref{fig:threelines}(i)), the sector of $e_i$ consists of two parts: (1)~The first is that part of $R$ that lies between the center ray of $r_j$ and the upper tangent of $r_p$. (2)~Let $R_m$ be the smallest ring enclosing $r_m$. The second part of the sector is the part of $R_m$ that lies between the center ray of $r_j$ and the upper tangent of $r_m$. 
		
		\textbf{The potential of $e_i$}: $r_j$ pays the half of its volume and $r_p,r_m$ their entire volumes to the sector of $e_i$.
		
		\item[T8] \textbf{The sector of $e_i$}: If $e_i = (r_j,r_m)$ is a vertical start and end edge,  (as shown in Fig.~\ref{fig:threelines}(j)), the sector of $e_i$ is consists of two parts: (1)~The first is that part of $R$ that lies between the minimum of the lower tangents of $r_j,r_p$ and the maximum of the upper tangents $r_j,r_p$. (2)~Let $R_m$ be the smallest ring enclosing $r_m$. The second part of the sector is that part of $R_m$ that lies between the lower and the upper tangent of $r_m$.
		
		\textbf{The potential of $e_i$}: $r_j,r_p,r_m$ pay their entire volume to the sector of $e_i$.  
	\end{itemize}
	
	For simplicity, we also call the density of the sector of an edge $e_i \in P$ the \emph{density of $e_i$}. The \emph{sector} of a zipper is the union of the sectors of the edges of $P$.
	
\begin{restatable}{lemma}{lemzipperoflengthatleastthree}\label{lemzipperoflengthatleastthree}
	Let $Z = \langle r_k,\dots,r_{\ell} \rangle$ be a zipper of length at least two and $P$ a path in the adjacency graph of $Z$ connecting $r_k$ with $r_{\ell}$. Each edge $e_i \in P$ has a density of at least $\rho$.
\end{restatable}

{The proof of Lemma~\ref{lemzipperoflengthatleastthree} is the only
computer-assisted proof. All remaining proofs are analytic.} Due to space
constraints, the proof of Lemma~\ref{lemzipperoflengthatleastthree} is given in
the appendix. Combining Lemmas~\ref{lem:zipperlengthone}
and~\ref{lemzipperoflengthatleastthree} yields the following.

\begin{corollary}\label{cor:densityzippersectors}
	Sectors of zippers have a density of at least $\rho$.
\end{corollary}

Ring Packing stops when the sum of the diameters of the current disk $r_i$ and the disk packed last $r_{i-1}$ is smaller than the width $w$ of the current ring, i.e., if $2r_{i-1} + 2r_i < w$. If $2r_{i-1} + 2r_i < w$, Phase 5 partitions the current ring into two new open rings with widths $2r_i, w-2r_i$. Hence, the sectors of zippers packed by Ring Packing become firmly interlocked without leaving any gaps between two zippers, see Fig.~\ref{fig:transit}.
\begin{figure}[]
  \begin{center}
      \includegraphics[scale=1]{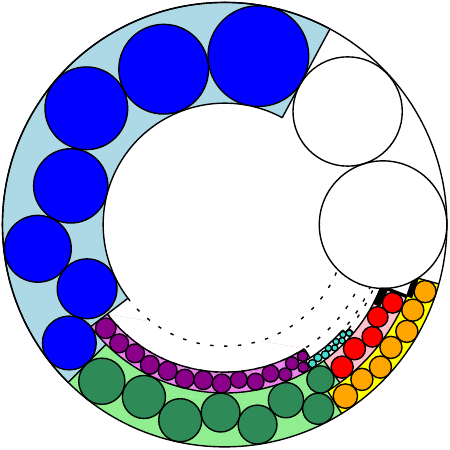} 
  \end{center}

  \caption{The sectors of rings packed by Ring Packing become firmly interlocked without leaving any gaps between two sectors. The minimal rings into which the orange and the red zippers are packed are full. The minimal ring into which the turquoise zipper is packed is open.}
  \label{fig:transit}
\end{figure}
The only sectors that we need to care about are the gaps that are left by Ring Packing due to the second break condition, i.e., the current disk does not fit into the current ring, see the black sectors in Fig.~\ref{fig:transit}.

\new{\begin{corollary}\label{cor:ringpartrho}
	Let $R$ be a minimal ring and $G$ its gap. $R \setminus G$ has a density of at least $\rho$.
\end{corollary}}

	In order to analyze the gaps left by Ring Packing, we first need to observe
for which rings we need to consider gaps. In particular, we have two break
conditions for Ring Packing:

	(1)~The current disk $r_i$ does not fit into the current ring $R$, causing us to close the ring and disregard it for the remainder of the algorithm?
	
	(2)~The current and the last disk $r_{i-1}$
packed into $R$ can pass one another, resulting in $R$ to be partitioned into several rings with
smaller widths. Thus, we obtain that two computed rings $R_1,R_2$ either do not overlap or $R_1$ lies inside $R_2$.
This motivates the following definition.

\begin{definition}\label{def:ringstructure}
Consider the set of all rings $R_1,\dots,R_k$ computed by our algorithm. A ring $R_i$ is \emph{maximal} if there is no ring $R_j$ with $R_i \subset R_j$. A ring $R_i$ is \emph{minimal} if there is no ring $R_j$ with $R_i \supset R_j$.
\end{definition}

\new{By construction of the algorithm,} each ring is partitioned into minimal rings. Thus, we define gaps only for minimal rings, see Figure~\ref{fig:gap} and Definition~\ref{def:gap}.

\begin{figure}[t]
  \begin{center}
      \includegraphics[scale=1]{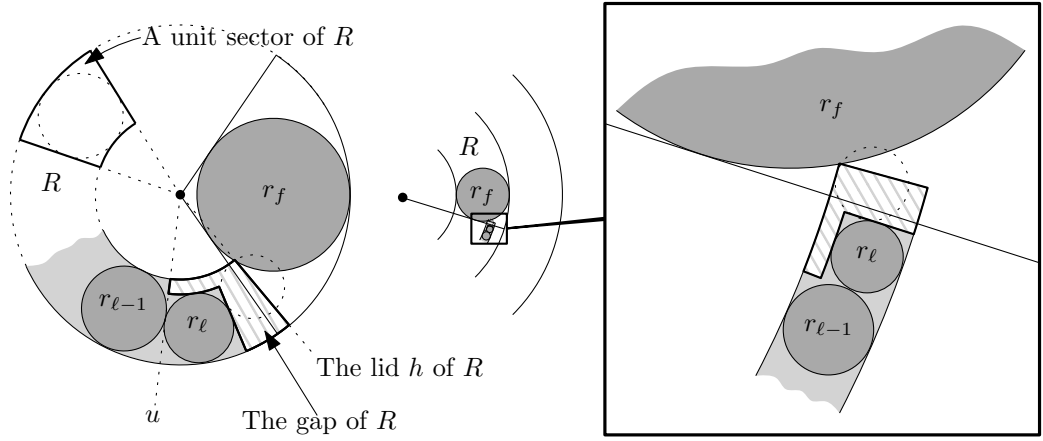} 
  \end{center}
  \vspace*{-12pt}
  \caption{The lid, the gap (striped white-gray), and a unit sector of a ring $R$.}
  \label{fig:gap}
\end{figure} 

\begin{definition}\label{def:gap}
	Let $Z = \langle \dots, r_{\ell-1}, r_{\ell} \rangle$ be a zipper of length at least $2$ inserted into a minimal ring $R$. The \emph{lid} \new{$h$} of $R$ is the ray above the upper tangent $u$ of $r_{\ell}$ such that $h$ realizes a maximal polar angle while $h \cap R$ does not intersect an already packed disk $r_f$ with $f \leq \ell-1$, see Fig.~\ref{fig:gap}. The \emph{gap} of $R$ is the part of $R$ between the upper tangent $u$ of $r_{\ell-1}$ and the lid of $R$ which is not covered by sectors of $Z$, see the white-gray striped sectors in Fig.~\ref{fig:gap}.
	
	A \emph{unit sector} of $R$ is a sector of $R$ that lies between the two tangents of a disk touching the inner and the outer boundary of $R$, see Fig.~\ref{fig:gap}. The \emph{unit volume} $U_R$ of $R$ is the volume of a unit sector of $R$.
\end{definition}

The lid of a gap lies either inside a cone induced by a disk packed by \new{Boundary} Packing, see Fig.~\ref{fig:gap}~(Left), or inside the sector of a zipper packed by Ring Packing, see Fig.~\ref{fig:gap}~(Right). This leads to the following observation

\begin{observation}\label{obs:minimalringcoveredbysectorsandgap}
	Each minimal ring $R$ is covered by the union of cones induced by disks packed by \new{Boundary} Packing into $R$, sectors of zippers packed by Ring Packing into $R$, and the gap of $R$.
\end{observation}

Next, we upper bound the volume of the gap of minimal rings.

\begin{lemma}\label{lem:upperboundlock}
	The gap of a minimal ring $R$ has a volume of at most $1.07024 U_R$.
\end{lemma}
\begin{proof}
	As we want to upper bound the volume of the gap w.r.t. the unit volume $U_R$ of $R$, w.l.o.g. we make the following assumptions (A1)-(A4), see Fig.~\ref{fig:upperbounding}:
	
	\begin{figure}[h!]
  \begin{center}
      \includegraphics[scale=1]{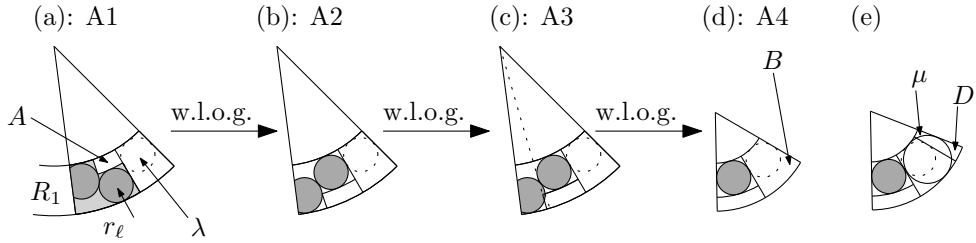} 
  \end{center}
  \vspace*{-12pt}
  \caption{Simplifying assumptions that do not increase the density.}
  \label{fig:upperbounding}
\end{figure}
	
	\begin{itemize}
		\item (A1) The largest disk $\lambda$ inside $R$ touching $h$ from below, the upper tangent of $r_{\ell}$ from above, and the inner boundary of $R$, such that $\lambda$ does not overlap with any other disks from below, has the same radius as $r_{\ell}$, see Fig.~\ref{fig:upperbounding}(a).
		\item (A2) The last disk $r_{\ell}$ packed into $R$ touches the inner boundary of $R$, see Fig.~\ref{fig:upperbounding}(b).
		\item (A3) The empty pocket $A$ left by the sector of the end edge of the zipper inside $R$ is bounded from below by the lower tangent of $r_{\ell}$ but not by the upper tangent of $r_{\ell-1}$, see Fig.~\ref{fig:upperbounding}(c).
		\item (A4) $r_{\text{out}}= 1, r_{\text{in}} = \frac{1}{2}$, see Fig.~\ref{fig:upperbounding}(d).
	\end{itemize}

	Let $B$ be the sector of $R$ that lies between the two tangents of \new{$\lambda$}, see Fig.~\ref{fig:upperbounding}\new{(d)}. We upper bound the volume of the gap of $R$ as \new{$\abs{A} + \abs{B} \leq 1.07024 U_R$}, as follows.
	
	Let $\mu \subset R$ be the disk touching the inner and the outer boundary of $R_1$ and the upper tangent of $r_{\ell}$ from above, see Fig.~\ref{fig:upperbounding}(e). Furthermore, let $D$ be the part of the cone induced by $\mu$ \new{which} lies inside $R$ and between the upper and lower tangent of $\mu$, see Fig.~\ref{fig:upperbounding}(e). 
	
	In the following, we show that $\abs{A} - \abs{D} \leq 0.07024 U_R$.
	
	\begin{eqnarray*}
		 \abs{A} - \abs{D}& \leq &\phantom{-} \frac{2 \arcsin \left(\frac{\lambda}{\frac{1}{2} + \lambda} \right)}{2\pi} \pi \left( 1 - \left( \frac{1}{2} + 2 \lambda \right)^2 \right) \\
		&&- \frac{2\arcsin \left( \frac{1}{3} \right) - 2\arcsin \left( \frac{\lambda}{\frac{1}{2} + \lambda} \right)}{2 \pi} \pi \left( \frac{3}{4} \right)\\
		&=&  \phantom{-}\arcsin \left(\frac{\lambda}{\frac{1}{2} + \lambda} \right) \left( \frac{7}{4} - \left( \frac{1}{2} + 2 \lambda \right)^2 \right) \\
		&&- \frac{3}{4}\arcsin \left( \frac{1}{3} \right) =: V_{AD}.
	\end{eqnarray*}
	%Put maximize arcsin(a/(0.5+a))(1-(0.5+2a)^2)+ arcsin(a/(0.5+a))3/4 on 1/8<=a<= 1/4 into wolframAlpha
	The first derivative of $V_{AD}$ is
	\begin{eqnarray*}
		\frac{d \; V_{AD} \; \lambda}{d \; \lambda} &=& \phantom{-} \frac{\left( \frac{1}{\frac{1}{2}+ \lambda} - \frac{\lambda}{\left(\frac{1}{2} + \lambda\right)^2} \right) \left( \frac{7}{4} - \left( 2 \lambda + \frac{1}{2} \right)^2 \right)}{\sqrt{1 - \frac{\lambda^2}{\left(\frac{1}{2} + \lambda\right)^2}}}\\
		&&- 4 \arcsin \left( \frac{\lambda}{\frac{1}{2} + \lambda} \right) \left( 2\lambda + \frac{1}{2} \right).
	\end{eqnarray*}
	
	Solving $\frac{d \; V_{AD} \; \lambda}{d \; \lambda} = 0$ yields $\lambda \approx 0.196638$. Finally, we observe that \new{$V_{AD}\left(\frac{1}{8} \right) \approx -0.01576$}, \new{$V_{AD}\left( 0.196638 \right) \approx 0.01756$}, \new{$V_{AD}\left(\frac{1}{4} \right) = 0$}. This implies that $\abs{A} - \abs{D} \leq 0.01756 \leq 0.07024 U_R$, because $U_R \geq \frac{1}{4}$. 
\end{proof}
		
\subsection{Analysis of the Algorithm for the Case $r_1\leq 0.495$}\label{sec:ronenothuge}

We show that each computed minimal ring is saturated, see Corollary~\ref{cor:eachminimalringsaturated}. Let $R_1,\dots,R_h \subseteq \mathcal{C}$ be the created minimal rings ordered decreasingly w.r.t.\ their outer radii. The inner boundary of $R_i$ is the outer boundary of $R_{i+1}$ for $i = 1,\dots,h-1$.

	We show by induction over $h$ that $R:=R[r_{\text{out}}, r_{\text{in}}]:=R_h$ is saturated. Thus, we assume that $R_1,\dots,R_{h-1}$ are saturated, implying that $\mathcal{C} \setminus r_{\text{out}}$ is saturated, where $r_{\text{out}}$ is the outer radius of $R_h$.

		For the remainder of Section~\ref{sec:ronenothuge}, each disk $r_i$ packed by \new{Boundary} Packing pays its entire volume to the cone induced by $r_i$.

%If $r_i \leq 0.495$ we can ensure a virtual packing density of at least $0.56127$.

%Thus, for the remainder of Section~\ref{sec:ronenothuge} we assume $r_n < \frac{1}{4}$.
	
	\begin{lemma}\label{lem:atleastonesmall}
		%If $r_1 \leq 0.495$
		Assume $r_n < \frac{1}{4}$. There is at least one disk $r_k$ packed into $R$ and touching both the inner and the outer boundary of $R$.
	\end{lemma}
	\begin{proof}
		Assume that our algorithm did not pack a disk with radius smaller than $\frac{1}{4}$ adjacent to~$\mathcal{C}$. Let $r_k$ be the largest disk not packed adjacent to~$\mathcal{C}$ \new{into $R$}. 
		
		By Lemma~\ref{lem:largestdisknotadjacenttoboundary}, we obtain that $r_k$ is smaller than $\frac{1}{4}$. This implies that the volume of the sector that is not covered by the cones induced by $r_1,\dots,r_{k-1}$ is upper bounded by $\arcsin \left( \frac{1}{3} \right)$, see Fig.~\ref{fig:atleastonesmalldiskadjacentoboundary}.
		
\begin{figure}[h!]
  \begin{center}
      \includegraphics[scale=1]{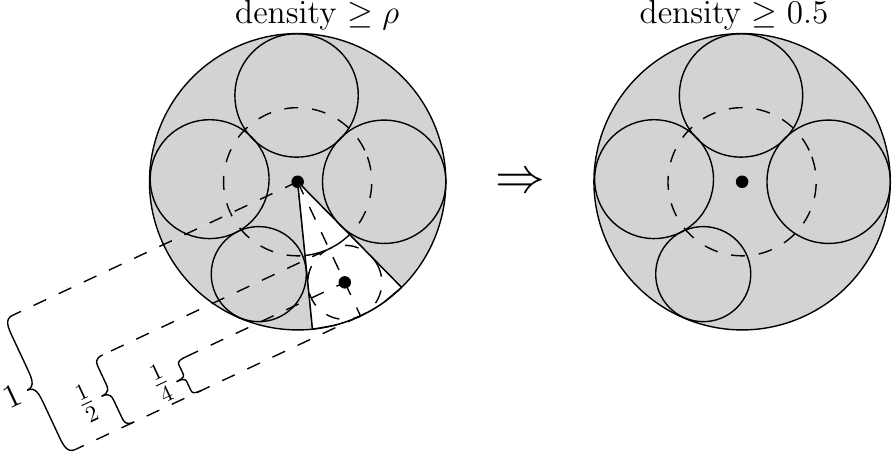} 
  \end{center}

  \caption{Ensuring density of at least $\rho$ for all cones induced by disks packed by \new{Boundary} Packing implies a density of at least $0.5$ for the entire container disk.}
  \label{fig:atleastonesmalldiskadjacentoboundary}
\end{figure}

	Each disk $r_i$ from $r_1,\dots,r_{k-1}$ pays its entire volume to the cone induced by~$r_i$. Lemma~\ref{lem:densityCones} implies that each cone has a density of at least $\rho$, because $r_1,\dots,r_n \leq 0.495$. This implies that the total volume of $r_1,\dots,r_{k-1}$ is at least $\pi \cdot \rho \cdot \frac{2\pi - 2\arcsin \left( 1/3 \right)}{2\pi} = \rho (\pi - \arcsin \left( 1/3 \right)) > \frac{\pi}{2}$ \new{contradicting the assumption that the total input volume is no larger than $\frac{\pi}{2}$}.
	\end{proof}
	
\begin{lemma}\label{lem:r1smalldenserings}
	$R_h$ is saturated.
\end{lemma}
\begin{proof}
	Let $S_1$ be the sector of $R_h$ that is covered by cones induced by disks packed by \new{Boundary} Packing or by sectors of zippers packed by Ring Packing. Lemma~\ref{lem:atleastonesmall} implies that there is a disk $r_k$ packed into $R_h$ such that $r_k$ touches the inner and the outer boundary of $R_h$. Let $S_2$ be the sector of $R_h$ between the lower and the upper tangent of $r_k$.
	
	We move potentials $\delta_1, \delta_2$ from $S_1,S_2$ to a potential \new{variable} $\Delta$ and guarantee that $\Delta$ is at least $\frac{1}{2}$ times the volume of the gap $G$ of $R_h$. Finally, we move $\Delta$ to $G$, implying that $G$ is saturated, which in turn implies that $R_h$ is saturated.
	
	Lemma~\ref{lem:zipperlengthone} implies that the density of $S_2$ is at least $0.77036$. We move a potential $\delta_2 := \left( 0.77036 - \rho \right) |S_1| > 0.20971 U_{R_h}$ from $S_2$ to $\Delta$, implying that $S_2$ has still a density of $\rho$.
	
	Combining Lemma~\ref{lem:densityCones} and Corollary~\ref{cor:densityzippersectors} yields that $S_1$ has a density of at least $\rho$. Lemma~\ref{lem:upperboundlock} implies that the volume of the gap of $R_h$ is at most $1.07024 U_R$. The volume of $R_h$ is at least $\frac{2 \pi}{2 \arcsin \left( \frac{1}{3} \right)} U_{R_h} > 9.24441U_{R_h}$. Thus, the volume of $S_1$ is at least $(9.24441 - 1.07024) U_{R_h} = 8.17417U_{R_h}$. Hence, we move a potential $\delta_1:=\left( \rho -  \frac{1}{2} \right) 8.17417U_{R_h} > 0.49576 U_{R_h}$ to $\Delta$.
	
	We have $\Delta = \delta_1 + \delta_2 > 0.49576 + 0.20971 = 0.70547$, which is large enough to saturate a sector of volume $V_{\Delta} = 2 \cdot 0.70547= 1.41094 U_{R_h}$.
	
	As $\abs{G} \leq 1.07024$, moving $\Delta$ to $G$ yields that $G$ is saturated, which implies that $R_h$ is saturated. This concludes the proof.
\end{proof}

\begin{corollary}\label{cor:eachminimalringsaturated}
	Each minimal ring is saturated.
\end{corollary}

	As each ring can be partitioned into minimal rings, we obtain the following.
	
	\begin{corollary}\label{cor:allringssaturated}
		All rings are saturated.
	\end{corollary}

Combining Lemma~\ref{lem:largestdisknotadjacenttoboundary} and Corollary~\ref{cor:allringssaturated} yields that all disks are packed.
	
	\begin{lemma}\label{lem:alldisksfit}
		Our algorithm packs all input disks.
	\end{lemma}
	\begin{proof}
		By induction assumption we know that $O \setminus \mathcal{C}$ is saturated and Corollary~\ref{cor:allringssaturated} implies that all rings inside $\mathcal{C}$ are also saturated.
		
		Let $\overline{\mathcal{C}}$ be the disk left after removing all  rings from $\mathcal{C}$, implying that $\overline{\mathcal{C}}$ is empty. Lemma~\ref{lem:largestdisknotadjacenttoboundary} implies that a final iteration of \new{Boundary} Packing to $\overline{\mathcal{C}}$ yields that all remaining disks are packed into $\overline{\mathcal{C}}$. This concludes the proof.
	\end{proof}

\subsection{Analysis of the Algorithm for the \new{Case} $0.495 \leq r_1$}\label{sec:r1middle}
{
In this section we prove that all disks are packed if $0.495 \leq r_1$ by distinguishing whether $0.495 \leq r_1 \leq \frac{1}{2}$ or $\frac{1}{2} < r_1$.

If $0.495 \leq r_1 \leq \frac{1}{2}$, we apply a similar approach as used for the case $r_1 \leq 0.495$. The additional difficulty for the case of $0.495 \leq r_1 \leq \frac{1}{2}$ is that the cone induced by $r_1$ may have a density of $\frac{1}{2}$. Thus, we have to generate some extra potential from the remaining sectors in order to ensure that the gaps of the rings are saturated, see Section~\ref{sec:appendixanalysisr1large} for details.

\begin{restatable}{lemma}{lemalldisksfittwo}\label{lem:alldisksfit2}
		If $0.495\le r_1 \leq \frac{1}{2}$, our algorithm packs all disks into the container disk.
	\end{restatable}

If $\frac{1}{2} < r_1$, we need to refine our analysis because the midpoint of the container disk $\mathcal{C}$ lies inside $r_1$. In particular, we consider a \emph{half disk} $H$ lying inside $\mathcal{C}$ such that $H$ and $r_1$ are touching each other. The volume of $H$ is at least twice the volume of the remaining disks to be packed, see Figure~\ref{fig:reducingr1hugetor1largetwo}. Finally, applying a similar approach as used in the case of $0.495 \leq r_1 \leq \frac{1}{2}$ to $H$ yields that all disks are packed, see Section~\ref{sec:appendixanalysisr1huge} for details.

\begin{figure}[h!]
  \begin{center}
      \includegraphics[scale=1]{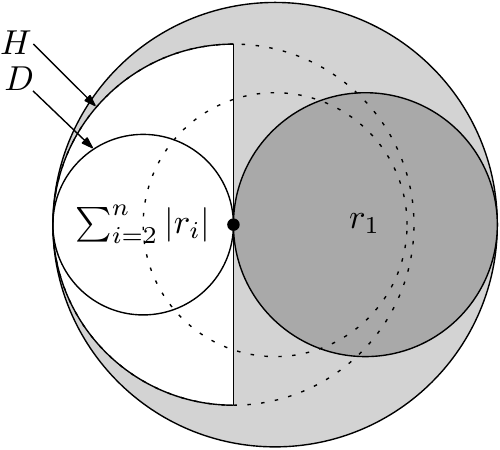} 
  \end{center}
  \caption{The total volume of the remaining disks to be packed is smaller than the volume of the white disk $D$. As $\abs{H} = 2 \abs{D}$, it suffices to guarantee that $H$ is saturated.}
  \label{fig:reducingr1hugetor1largetwo}
\end{figure}

\begin{restatable}{lemma}{lemalldisksfitthree}\label{lem:alldisksfitthree}
	If $\frac{1}{2} < r_1$, our algorithm packs all \new{disks} into the original container disk.
\end{restatable}

Lemma~\ref{lem:alldisksfitthree} concludes the proof of Theorem~\ref{thm:disk_packing}.
}

%% file: 05-hardness.tex
It is straightforward to see that the hardness proof for packing disks into a square can
be adapted to packing disks into a disk, as follows.

\begin{theorem}
It is NP-hard to decide whether a given set of disks fits into a circular container.
\end{theorem}

The proof is completely analogous to the one by Demaine, Fekete, and Lang in 2010~\cite{DFL2010circle}, 
who used a reduction from \textsc{3-Partition}.
Their proof constructs a disk instance which first forces some symmetrical free “pockets” in the 
resulting disk packing. The instance's remaining disks can then be packed into these pockets if and only 
if the related \textsc{3-Partition} instance has a solution. Similar to their construction, we 
construct a symmetric triangular pocket by 
using a set of three identical disks of radius $\frac{\sqrt{3}}{2+\sqrt{3}}$ that can only be packed into a unit disk
by touching each other. Analogous to \cite{DFL2010circle}, this is further
subdivided into a sufficiently large set of identical pockets. The remaining disks
encode a \textsc{3-Partition} instance that can be solved if and only if the disks can be 
partitioned into triples of disks that fit into these pockets.

\begin{figure}[h!]
  \begin{center}
      \includegraphics[scale=0.9]{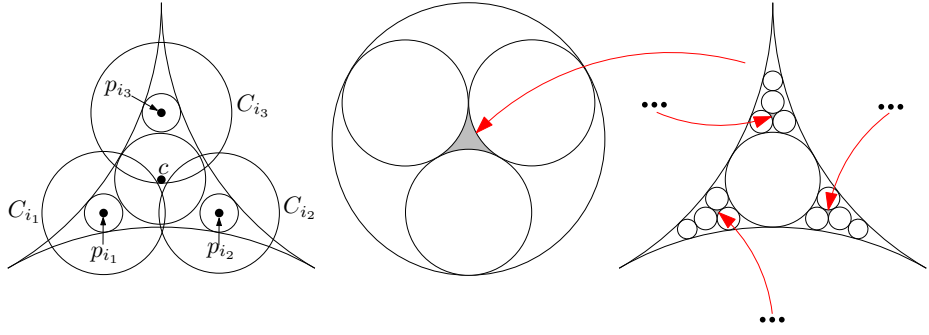} 
  \end{center}
  \caption{Elements of the hardness proof: (1) A symmetric triangular pocket from \cite{DFL2010circle}, allowing
three disks with centers $p_{i_1}$, $p_{i_2}$, $p_{i_3}$ to be packed if and only if the sum
of the three corresponding numbers from the \textsc{3-Partition} instance is small enough.
(2) Creating a symmetric triangular pocket in the center by packing three disks of radius $\frac{\sqrt{3}}{2+\sqrt{3}}$ and the adapted argument from \cite{DFL2010circle} for creating a sufficiently large set of symmetric
triangular pockets.}
  \label{fig:NPhard}
\end{figure}

%% file: 04A-analysis.tex
\section{Details of the Analysis of Ring Packing}\label{sec:detailsanalysis}

\restatethm{\lemzipperoflengthatleastthree*}{lemzipperoflengthatleastthree}

\begin{proof} In order to prove the lemma, we apply \emph{interval arithmetic}.
	In interval arithmetic, mathematical operations such as addition, multiplication or computing the sine are performed on intervals instead of real numbers.
	Operations on intervals are defined based on their real counterparts as follows.
	For two intervals $A=[a_1,b_1],B=[a_2,b_2]$ and some binary operation $\circ$, the result $A \circ B$ is defined as
	\[A\circ B \coloneqq \left[\min\limits_{x_1 \in A, x_2 \in B} x_1 \circ x_2, \max\limits_{x_1 \in A, x_2 \in B} x_1 \circ x_2\right].\]
	In other words, the result of an operation is the smallest interval containing all values $x \circ y$ for $x \in A, y \in B$.
	Unary operations are defined analogously.
	If the input interval(s) contain values for which the corresponding operation on real numbers is undefined, the result is undefined.
	
	In order to use interval arithmetic for our proof, we consider the cases T1--T8 as depicted in Figure~\ref{fig:threelines} separately.
	For each of these cases, we consider the following 3--4 variables.
	The first variable is $\lambda$, the inner radius of the ring; we assume the outer radius to be $1$.
	Additionally, we have 2--3 variables $r_1 \geq r_2 \geq r_3$ corresponding to the radii of the disks involved in the case.
	For each of the cases, it is straightforward to implement an algorithm that computes the density of any configuration T1--T8, given some real values for $\lambda,r_1,\ldots,r_3$.
	Such an algorithm needs to perform basic arithmetic operations as well as square root and inverse sine computations.
	Instead of implementing such an algorithm using real numbers, we can also implement it using interval arithmetic.
	As input, instead of concrete real values, we are given intervals $I_{\lambda},I_1,I_2,I_3$ for $\lambda$ and $r_1,r_2,r_3$.
	As output, we compute an interval $I_d$ for the density of any given configuration from T1--T8.
	We know that $I_d$ contains all possible density values that an implementation using real numbers can produce given inputs from $I_{\lambda},I_1,I_2,I_3$.
	Therefore, if the lower bound of $I_d$ is above a lower bound $b_d \coloneqq 0.5642$, we know that the density is at least $b_d$ for all possible values these intervals.

	Furthermore, we can bound our variables as follows.
	We only have to consider the case $\frac{1}{2} \leq \lambda < 1$.
	For any given value of $\lambda$, we know that $r_1 \leq \frac{1-\lambda}{2}$, because otherwise, $r_1$ would not fit into the ring.
	Moreover, we can lower-bound $r_2$ by $\frac{1-\lambda - 2r_1}{2}$, because otherwise, $r_2$ would be the first disk in a new ring; similar statements hold for $r_3$ in the three-disk cases.
	This gives us lower and upper bounds for all involved variables.
	By subdividing the ranges for $\lambda,r_1,r_2,r_3$ into sufficiently small intervals, we subdivide the space spanned by $\lambda,r_1,r_2,r_3$ into finitely many hypercuboids; these hypercuboids cover the entire space.
	
	We can feed each hypercuboid into the interval arithmetic implementation of each of the configurations T1--T8.
	If, for a hypercuboid, the result is a density interval with lower bound of at least $b_d$, we do not have to consider that hypercuboid anymore.
	Implementing this idea yields an automatic prover\footnote{The source code of this prover is available online:\\ \url{https://github.com/phillip-keldenich/circlepacking}} which we can use to prove that the density is at least $b_d$ for all hypercuboids with $\lambda \leq 0.99$.
	
	The number of hypercuboids that we have to consider in this process is large; therefore, we implemented the approach outlined above using a CUDA-capable GPU in order to perform the computations for individual hypercuboids in a massively parallel fashion.
	Our implementation of interval arithmetic handles rounding errors arising from the use of limited-precision floating-point arithmetic by using appropriate rounding modes (where available) or error bounds guaranteed by the underlying platform.
	In this way, we can ensure that the interval resulting from any operation contains all possible results of the corresponding operation on real numbers given values from the input intervals.
	This ensures soundness of our results in the presence of rounding errors.

	Let $R := R[r_{\text{out}}, r_{\text{in}}]$ be the ring into which $Z$ is packed. Our prover shows that the density of $e_i$ is $0.5642$, if the ratio of~$R$ is at most $0.99$. Hence, we assume w.l.o.g. that the ratio of \new{$r_{\text{in}}$ and $r_{\text{out}}$} is at least $0.99$. Furthermore, we assume w.l.o.g. that $R$ has a width of $0.01$.
	
	Before we lower bound the density of $e_i$, we need to define some notations. The \emph{midpoint} of a ring $R$ is the midpoint of the disk induced by the outer boundary. Let $q$ be a point in $R$ and $\ell_q$ the ray starting from \new{the center $m$ of the container disk} and shooting into the direction of $q$, see Figure~\ref{fig:coordinates}. Let $\overline{q}$ be the intersection point of $\ell_p$ with the outer boundary of $R$. Given a ray $\ell$, called \emph{reference axis}, the \emph{first coordinate} of $q$ is the distance between $q$ and $\overline{q}$. The \emph{second coordinate} of $q$ is \new{the} length of the curve $\beta$ on the outer boundary between $\ell$ and $\ell_q$. 
	
	\new{We scale $R$ and the disks' radii so that the reference axes match}. In particular, we denote the sector of $e_i$ by $B$. Let $(\overline{r}_j,\overline{r}_m) =: \overline{e}_i$ be the edge with $\overline{r}_m = r_m, \overline{r}_j := r_j$ packed by Ring Packing into a ring $\overline{R} := R[1,0.99]$ and $\overline{B}$ the sector of $\overline{e}_i$ such that $e_i$ and $\overline{e}_i$ are of the same type. We construct from $\overline{B} \subset R[1,0.99]$ a sector $A \subset R$ such that $A$ is a superset of the sector $B$ of $e_i$. We guarantee $\frac{\abs{A}}{\abs{\overline{B}}} \leq 1.00503$ implying that $B$ has a density of at least $\frac{0.5642}{1.00503} > 0.56137 > \rho$ because $|\overline{B}|$ has a density of at least $0.5642$.
	
\begin{figure}[h!]
  \begin{center}
      \includegraphics[scale=1]{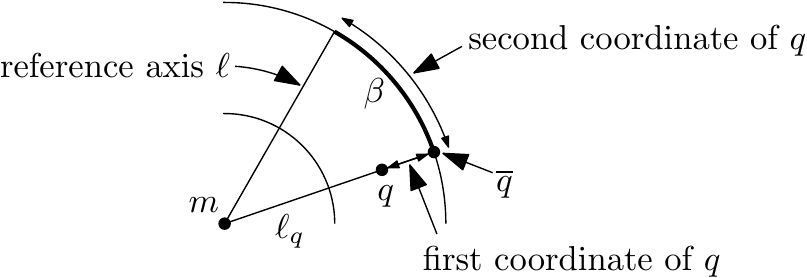} 
  \end{center}
  
  \caption{Regarding the reference axis $\ell$, the coordinates of the point $q$ are the length of $\beta$ and the distance between $q$ and $\overline{q}$.}
  \label{fig:coordinates}
\end{figure}

	Let $\overline{e}_i$ be of type T1, see Figure~\ref{fig:prover_remainingcase}(a).
\begin{figure}[h!]
  \begin{center}
      \includegraphics[scale=1]{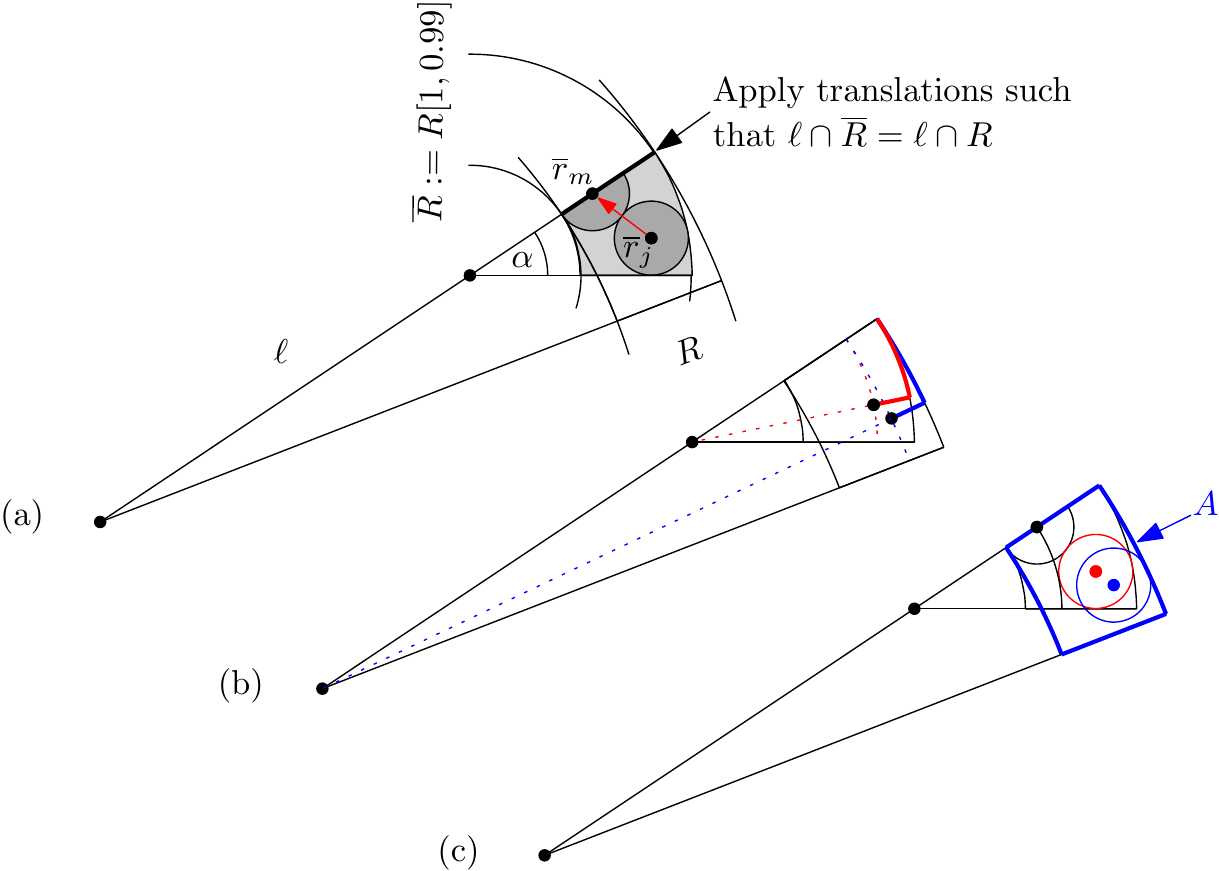} 
  \end{center}
  
  \caption{Mapping the midpoint of $r_j$ onto a point in $R$ having the same coordinates as the midpoint of $r_j$.}
  \label{fig:prover_remainingcase}
\end{figure}
W.l.o.g., we assume that $r_m$ touches the inner boundary of $\overline{R}$. Let $\ell$ be the ray shooting from the midpoint of $\overline{R}$ into the direction of $r_m$. We scale and place $R,r_j,r_m$ such that $\ell \cap \overline{R}$ is equal to $\ell \cap R$, see Figure~\ref{fig:prover_remainingcase}~(Left). Thus, both $R$ and $\overline{R}$ have the width $\frac{1}{100}$. We use $\ell$ as reference axis for both rings $R$ and $\overline{R}$. We map each point from the sector of $\overline{e}_i$ onto the point in $R$ having the same coordinates, see Figure~\ref{fig:prover_remainingcase}(b). Let $\overline{m}_m, \overline{m}_j$ be the midpoints of $\overline{r}_m, \overline{r}_j$. By construction, $\overline{m}_m$ is mapped onto itself. Furthermore, $\overline{m}_j$ is mapped onto a point that lies farther away from $\overline{m}_m$ as $\overline{m}_j$ does, because $\overline{m}_j$ lies closer to the outer boundary as $\overline{m}_m$ in the initial configuration.
	
	Let $A$ be the union of all points from $R$ onto which points from the sector of $e_i$ are mapped, see the blue sector in Figure~\ref{fig:prover_remainingcase}(c). The volume $\abs{A}$ of $A$ achieves its supremum for $r_{\text{in}}$ approaching $1$. Thus, we upper bound $\abs{A} \leq \frac{\beta}{100}$, where $\beta$ is the length of the part of the outer boundary of $\overline{R}$ inside the sector of $e_i$. As the outer radius of $\overline{R}$ is one, $\beta$ is also the angle induced by the sector of $e_i$, see Figure~\ref{fig:prover_remainingcase}(a). Furthermore, the volume \new{$\abs{\overline{B}}$} of the sector \new{$\overline{B}$} of \new{$\overline{e}_i$} is equal to $\frac{\beta}{2 \pi} \cdot \pi \cdot \left( 1 - \left( 1 - \frac{1}{100} \right)^2 \right)$. Thus, $\frac{\abs{A}}{\abs{\overline{B}}} \leq \frac{1}{1 - \frac{1}{200}} < 1.00503$. 

	The same approach as used for edges of type T1 applies to edges of type T2.
	
	\begin{figure}[t]
  \begin{center}
      \includegraphics[scale=1]{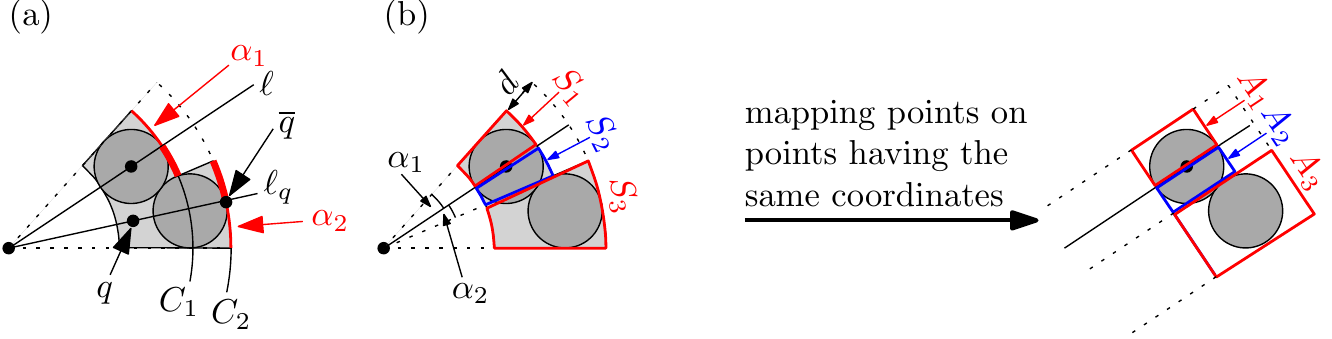} 
  \end{center}
  
  \caption{The coordinates of $q$ are the distance between $q$ and $\overline{q}$ and the sum of the lengths of the parts of $\alpha_1,\alpha_2$ that lie between $\ell$ and $\ell_q$.}
  \label{fig:coordinates2}
\end{figure}
	
	In order to analyze edges of type T3 and T4, we need a refined approach. We maintain the first coordinate of $q$ as the distance between $q$ and $\overline{q}$ and refine the second coordinate as follows: W.l.o.g., let $\overline{r}_m$ be the disk touching the inner boundary of $\overline{R}$. Let $\ell$ be the ray shooting from the midpoint of $\overline{R}$ into the direction of $\overline{m}_m$, see Figure~\ref{fig:coordinates2}. Let $C_1, C_2$ be the circles having the same midpoint as $\overline{R}$ such that $r_m,r_j$ touch $C_1,C_2$ from the interior. Let $\alpha_1 \subset C_1,\alpha_2 \subset C_2$ be the circular arcs lying on the boundary of $\overline{B}$. The second coordinate of $q$ is the sum of the lengths of the parts of $\alpha_1, \alpha_2$ between $\ell_q$ and $\ell$, see the fat red curves in Figure~\ref{fig:coordinates2}.
	
	Using the refined definition of coordinates, we apply the approach used for edges of type T1 and T2 to that part of $\overline{B}$ that lies below $\ell$. Furthermore, we apply the symmetric approach reflected at $\ell$ to that part of $\overline{B}$ that lies above $\ell$. The same argument as used for edges of type T1 and T2 implies that the distance between $\overline{m}_j, \overline{m}_m$ is not decreased.
	
	In order to upper bound the volume $\abs{A}$ of $A$, we partition $\overline{B}$ into three sectors $S_1, S_2,S_3$, see Figure~\ref{fig:coordinates2}(b). $S_1$ is that part of $\overline{B}$ which lies above $\ell$. $S_2$ that part of $\overline{B}$ which lies above the upper tangent to $r_m$ and below $\ell$. $S_3$ is $\overline{B} \setminus S_1 \cup S_3$. We consider the width $d$ of the pocket left by $\overline{B}$, see Figure~\ref{fig:coordinates2}(b). Furthermore, let $A_1,A_2,A_3$ be the images of $S_1,S_2,S_3$ under mapping points from $\overline{R}$ onto points from $R$. The volumes $\abs{A_1}, \abs{A_2}, \abs{A_3} $ of $A_1,A_2,A_3$ achieve their suprema for $r_{\text{in}}$ approaching $1$. Hence, we upper bound $|A_1|$ by $\frac{\alpha_1}{2 \pi} 2 \pi \frac{1-d}{100}$, where $\alpha_1$ is the angle induced by $S_1$, see Figure~\ref{fig:coordinates2}(b). Furthermore, we have $|S_1|=\frac{\alpha_1}{2 \pi} \pi  \left( (1-d)^2 - \left( (1-d) - \frac{1}{100} \right)^2 \right)$. Thus, we upper bound $\frac{\abs{A_1}}{\abs{S_1}}$ by $\frac{1}{1 - \frac{\frac{1}{100} - d}{2(1-d)}} \leq \frac{1}{1 - \frac{1}{200}} < 1.00503$.
	
	Using the same approach as used to upper bound $\frac{|A_1|}{|S_1|}$, yields $\frac{\abs{A_2}}{\abs{S_2}} < 1.00503$.
	
	Using the same approach as used for edges of type T1 we obtain $\frac{\abs{A_3}}{\abs{S_3}} \leq 1.00503$. This implies $\frac{\abs{A}}{\abs{B}} \leq 1.00503$ because $A = A_1 \cup A_2 \cup A_3$ and $B = S_1 \cup S_2 \cup S_3$ \new{where $A \cap B = \emptyset$}.
	
	 If $e_i = (r_j,r_m)$ of type T5, we assume w.l.o.g. that the disk $r_p$ packed between $r_j$ and $r_m$ lies adjacent to the inner boundary of $R$. Otherwise, the area of the sector of $e_i$ is monotonically decreasing in the ratio of $R$ implying that the sector of $e_i$ has a density of at least $0.5642$. Let $m_p$ be the midpoint of $r_p$ and $\ell$ the ray shooting from the midpoint of $R$ into the direction of the $m_p$. We apply the same approach as used edges of T1 to that part of $B$ which lies above $\ell$ and the symmetric approach to that part of $B$ which lies below $\ell$. This yields a lower bound of $0.56137$ for the density of $e_i$. 
	 
	 If $e_i$ is of type T6, applying the same approach as used for edges of type T5 yields $0.56137$ as a lower bound for the density of $e_i$.
	 
	 If $e_i$ is of type T7 \new{or of type T8}, applying the same approach as used for edges of type T3 yields that the density of $e_i$ is lower-bounded by $0.56137$.
	 
	 This concludes the proof.
\end{proof}

\section{Details of the Analysis for the Case of $0.495 \leq r_1 \leq \frac{1}{2}$}\label{sec:appendixanalysisr1large}

In this section, we show that our algorithm packs all disks if $0.495 \leq r_1 \leq \frac{1}{2}$.
	
	Let $R_1,\dots,R_{h}$ be the maximal rings ordered decreasingly w.r.t.\ 
their outer radii. As $R_1,\dots,R_{h}$ are maximal, $R_1,\dots,R_{h}$ are also
ordered decreasingly w.r.t.\ their widths, because our algorithm processes the
disks $r_1,\dots,r_n$ in decreasing order. The inner radius of $R_i$ is equal
to the outer radius of $R_{i+1}$ for $i = 1,\dots,h$.
	
	We distinguish whether there is a maximal ring with inner radius smaller than $\frac{1}{2}$ or not.
	
	\begin{lemma}\label{lem:allmaximalringssmall1}
		If all inner radii of maximal rings are \new{larger} than $\frac{1}{2}$, all disks are packed.
	\end{lemma}
	\begin{proof}
		Lemma~\ref{lem:halfcirclefitsinpocket} implies that the center of the smallest disk inside each minimal ring $R_{\min}$ that lies inside $R_1,\dots,R_h$ lies above the upper tangent of $r_1$. Thus, Lemma~\ref{lemzipperoflengthatleastthree} implies that all rings of $R_1,\dots,R_h$ are saturated. A final application of Lemma~\ref{lem:largestdisknotadjacenttoboundary} implies that all disks are packed.
	\end{proof}
	
	Hence, we assume w.l.o.g.\  that there is a maximal ring with inner radius smaller than $\frac{1}{2}$.

	Let $R_{m}$ be the maximal ring with outer radius \new{r_m} larger and inner radius not larger than $\frac{3}{4}$.
	
	First, we show that we can move a potential of $\frac{1}{4} \arcsin \left( \frac{1}{3} \right)r_m^2$ from $R_1,\dots,R_m$ to potential function $\Delta$ while guaranteeing that $R_1, \dots,R_m$ is saturated, see Corollary~\ref{cor:removingpotentials}. We distinguish whether $m = 1$ or $m > 1$, see Lemma~\ref{lem:technical} and Lemma~\ref{lem:removepotentialssmallrings}.
	
	For the remainder of Section~\ref{sec:r1middle}, each disk $r_i$ packed by \new{Boundary} Packing in Phase 1 pays its entire volume to the cone induced by $r_i$.

Let $r_m$ be the inner radius of $R_m$. 

\begin{lemma}\label{lem:technical}
	Let $0.495 \leq r_1 \leq \frac{1}{2}$ and $m = 1$. We can move a potential of $\frac{1}{4} \arcsin \left( \frac{1}{3} \right)r_m^2$ from~$R_1$ to potential function $\Delta$ while guaranteeing that $R_1$ is saturated. 
\end{lemma}
\begin{proof} Let $C$ be the cone induced by $r_1$, see Fig.~\ref{fig:r1large}. Lemma~\ref{lem:densityCones} implies that $C$ is saturated. W.l.o.g., we assume $r_1 = 0.495$ and that $C$ has a density of $\frac{1}{2}$.
\begin{figure}[h!]
  \begin{center}
      \includegraphics[scale = 1]{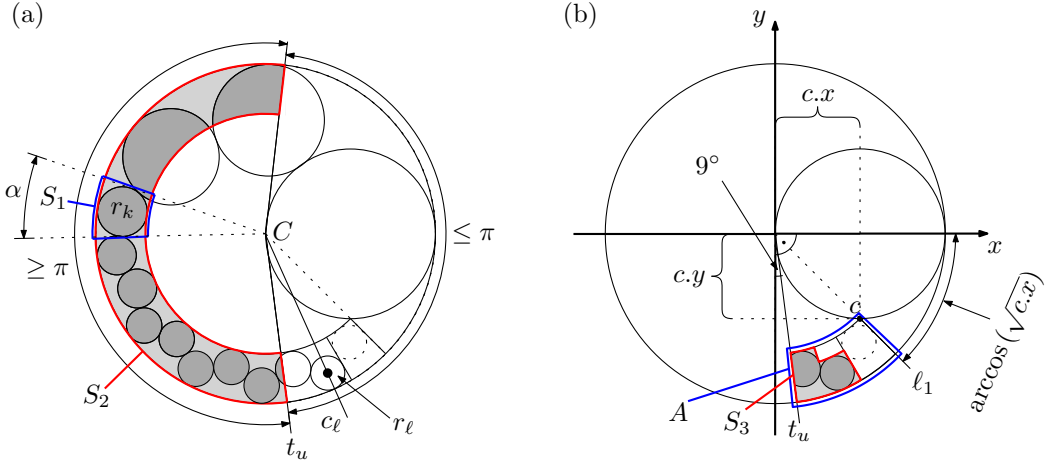} 
  \end{center}
  
  \caption{The case in which the largest maximal ring has a width at least as large as $\frac{1}{4}$.}
  \label{fig:r1large}
\end{figure}

	In the following, we consider three sectors $S_1,S_2,S_3 \subset R_1$. We prove that the potentials of $S_1,S_2,S_3$ can be reduced by certain potentials (denoted by $\delta_1,\delta_2,\delta_3$) while still ensuring that a density of $\frac{1}{2}$ for $R_1$. We iteratively remove $\delta_1,\delta_2,\delta_3$ from $S_1,S_2,S_3$ and pay it to $\Delta$. 

	As we want to lower bound the potential that we can remove from $R_1$ while maintaining that $R_1$ is saturated, we upper bound the sum of the gaps of all minimal rings inside $R_1$. Thus, we assume w.l.o.g.\  that $R_1$ is minimal, because this does not increase the sum of the gaps inside $R_1$. Lemma~\ref{lem:upperboundlock} implies that the gap of $R_1$ is at most $1.07024 U_{R_1}$.
	
	\begin{itemize}
		\item Construction of $S_1,\delta_1$: Lemma~\ref{lem:helping} implies that there is a disk $r_k$ packed into $R_1$ touching the inner and the outer boundary of $R_1$, see Fig.~\ref{fig:r1large}(a). Let $\alpha$ be the angle of the cone induced by $r_k$. Lemma~\ref{lem:zipperlengthone} implies that the part $S_1$ of $R_1$ that lies between the two tangents of $r_k$ has a density of at least $0.77036$. We remove a potential of $\delta_1 := \left( 0.77036 - \rho \right) |S_1| = 0.20936 |S_1| = 0.20936 U_{R_1}$ from $S_1$ and pay it to~$\Delta$. This implies that the density of $S_1$ is equal to $\rho$ and that we have saved a potential enough for achieving a density of $\rho$ for a sector of volume $\frac{0.20936}{\rho}U_{R_1} \approx 0.37319 U_{R_1}$.
	
		\item Construction of $S_2,\delta_2$: Lemma~\ref{lem:halfcirclefitsinpocket} implies that the midpoint of the last disk $r_{\ell}$ packed into~$R_1$ lies above the upper tangent $t_u$ of $r_1$, see Fig.~\ref{fig:r1large}(a). Combining Lemma~\ref{lemzipperoflengthatleastthree} and Lemma~\ref{lem:densityCones} implies that the part $S_2$ of $R_1$ that lies above the lower tangent of $r_1$ and below the center ray $c_{\ell}$ of $r_{\ell}$ has a density of $\rho$, see Fig.~\ref{fig:r1large}(a). The angle of the cone $C$ induced by $r_1$ is upper bounded by $\pi$, because $r_1 \leq \frac{1}{2}$. This implies that the angle between the upper tangent and the lower tangent of $r_1$ in counterclockwise order is lower bounded by $\pi$, see Fig.~\ref{fig:r1large}(a). Hence, the volume of $S_2$ is lower bounded by $\frac{\pi}{\arcsin \left( \frac{1}{3} \right)} U_{R_1} > 4.6222 U_{R_1}$. As~$S_2$ has a density of $\rho$, we move a potential of $\left( \rho - \frac{1}{2} \right)\abs{S_2} > \left( \rho - \frac{1}{2} \right) 4.6222 U_{R_1} > 0.28033 U_{R_1} =: \delta_2$ from $S_2$ to $\Delta$ while ensuring that $S_2$ is saturated and thus $R_1$ as well. The value of $\Delta$ is at least as large as the potential needed to achieve a density of $\rho$ for a sector of volume $(\frac{0.28033}{\rho} + 0.37319)\abs{S_1} >  (\frac{1}{2} + 0.37319)\abs{S_1} = (1-0.12681)\abs{S_1}$.
	 
	 	\item Construction of $S_3,\delta_3$: Let $S_3$ be the sector defined as that part of $C$ that is covered by sectors of the unique\footnote{There is only one zipper inside $R_1$, because we have assumed w.l.o.g.\ that $R_1$ is also minimal.} zipper inside $R_1$, see Fig.~\ref{fig:r1large}(b). We move the sum $\delta_3$ of the potentials payed by disks of the zipper to $S_3$ to $\Delta$. Lemma~\ref{lem:densityCones} implies that removing $\delta_3$ from~$S_3$ still ensures that $C$ is saturated.
	
		\item Lower-bounding $\Delta := \delta_1 + \delta_2 + \delta_3$: We lower bound the value of $\Delta$ in terms of that part~$A$ of $R_1$ that lies between the \new{lower} tangent $t_u$ of $r_1$ and the lid $\ell_1$ of $R_1$, see the blue sector in Fig.~\ref{fig:r1large}(b). In particular, we lower bound the volume $V_{\Delta}$ of an arbitrary sector such that $\frac{\Delta}{V_{\Delta}} \geq \rho$.
	
	Lemma~\ref{lem:upperboundlock} implies that the gap of $R_1$ has a volume of at most $1.07024 U_{R_1}$, implying $$V_{\Delta}\geq \abs{A} - 1.07024U_{R_1} + 0.87319 U_{R_1} = \abs{A} - 0.19705 U_{R_1}.$$
	
	Recall that $C$ is the cone induced by $r_1 = 0.495$.  Furthermore, we consider the intersection point $c = (c.x,c.y)$ of the boundary of $r_1$ and the inner boundary of $R_1$, where $c.x,c.y$ are the coordinates of $c$, see Fig.~\ref{fig:r1large}(b). As $c$ lies on the boundary of $r_1$, we have $\left(c.x - \frac{1}{2} \right)^2 + c.y^2 = \frac{1}{4}$, which is equivalent to $c.y^2 = c.x - c.x^2$. \new{Let $m$ be the center of the container disk}. Thus, $|mc| = \sqrt{c.x}$, implying that the angle $\beta$ between the center ray of $r_1$ and $\ell_1$ is $\arccos \left( \sqrt{c.x} \right)$.
	
	Furthermore, we obtain from $\abs{mc} = \sqrt{c.x}$ that the volume of $R_{1}$ is equal to $\pi \left( 1-c.x \right)$.
	
	We observe that the angle between the center ray of $r_1$ and the upper tangent of $r_1$ is $\frac{\pi}{2} - \arccos \left( \frac{r_1}{1-r_1} \right) \geq \frac{\pi}{2} - \arccos \left( \frac{0.495}{0.505} \right) \geq 0.4365 \pi$. This yields $\abs{A}\geq \frac{0.4365 \pi - \arccos (\sqrt{c.x})}{2 \pi} \pi (1-c.x)$.
	
	The angle $\alpha$ of the cone induced by $r_k$ is equal to $\arcsin \left( \frac{1- \sqrt{c.x}}{1+\sqrt{c.x}} \right)$, because $|mc| = \sqrt{c.x}$. Hence, we upper bound $\abs{S_1} \leq \frac{2\arcsin \left( \frac{1- \sqrt{c.x}}{1+\sqrt{c.x}} \right)}{2 \pi} \pi (1-c.x)$.
	
	We lower bound as follows. 
	
	\begin{eqnarray*}
		V_{\Delta} - \frac{1}{2}\arcsin \left( \frac{1}{3} \right) c.x &\geq & \phantom{-}|A| - 0.19705 \abs{S_1} - \frac{1}{2}\arcsin \left( \frac{1}{3} \right) c.x\\
			& \geq &\phantom{+} \frac{1}{2}\left( 0.4365 \pi - \arccos \left( \sqrt{c.x} \right) \right) \left( 1 - c.x \right)\\
			&&- 0.19705 \cdot \frac{1}{2}\left(2\arcsin \left( \frac{1- \sqrt{c.x}}{1+ \sqrt{c.x}} \right) \right) \left( 1 - c.x \right) \\
		&& - \frac{1}{2}\arcsin \left( \frac{1}{3} \right) c.x \eqqcolon f(c.x).
	\end{eqnarray*}
	Furthermore, we have $\sqrt{c.x} \in [\frac{1}{2}, \frac{3}{4}]$, because the width of $R_1$ is lower bounded by~$\frac{1}{4}$ and upper bounded by $\frac{1}{2}$.
	Thus, $c.x \in [\frac{1}{4},0.5625]$.

	We continue by proving $f(c.x) \geq 0$.
	Let \begin{align*}
		f_1(c.x) &\coloneqq \frac{1}{2} \cdot \left(0.4356\pi -\arccos\left(\sqrt{c.x}\right)\right)\cdot (1-c.x)\\
		f_2(c.x) &\coloneqq 0.19705 \cdot \arcsin\left(\frac{1-\sqrt{c.x}}{1+\sqrt{c.x}}\right)\cdot (1-c.x) + \frac{\arcsin(1/3)}{2}c.x,
	\end{align*}
	such that $f(c.x) = f_1(c.x) - f_2(c.x)$.
	Because $x$ is positive, we have \begin{align*}
		\frac{df_1(c.x)}{dc.x} &= -0.4365\pi + \overbrace{\frac{\sqrt{1-c.x}}{2\sqrt{c.x}}}^{\text{mon. decr.}} + \overbrace{\arccos\left(\sqrt{c.x}\right)}^{\text{mon. decr.}}\\
		\frac{df_2(c.x)}{dc.x} &= \frac{\arcsin(1/3)}{2} + 0.19705 \cdot \left(\frac{\sqrt{c.x}-1}{2c.x^{3/4}}+ \underbrace{\arcsin\left(1-\frac{2}{\sqrt{c.x}+1}\right)}_{\text{mon. incr.}}\right)\text{, and}\\
		\frac{\sqrt{c.x}-1}{2c.x^{3/4}} &= \frac{-1}{2c.x^{1/4}}\cdot\left(\frac{1}{\sqrt{c.x}}-1\right)\text{ is monotonically increasing due to }c.x<1.
	\end{align*}
	The first-order derivative of $f_1(c.x)$ is monotonically decreasing, and the first-order derivative of $f_2(c.x)$ is monotonically increasing on $c.x \in [\frac{1}{4},0.5625]$.
	Furthermore, we have \begin{alignat*}{3}
		\frac{df_1(c.x)}{dc.x}\Big(\frac{1}{4}\Big) &\approx 0.5447 > 0,&\ \frac{df_1(c.x)}{dc.x}\Big(0.5625\Big) &\approx -0.2048 < 0,\\
		\frac{df_2(c.x)}{dc.x}\Big(\frac{1}{4}\Big) &\approx -0.0364 < 0,&\ \frac{df_2(c.x)}{dc.x}\Big(0.5625\Big) &\approx 0.1037 > 0,
	\end{alignat*}
	which together with the monotonicity of the derivatives means that $f_1$ attains its minimum and $f_2$ attains its maximum on $c.x \in \left[\frac{1}{4},0.5625\right]$ at either $c.x = \frac{1}{4}$ or $c.x = 0.5625$.
	Therefore, we find that $\min f(c.x) \geq \min f_1(c.x) - \max f_2(c.x) \approx 0.1215 - 0.1079 > 0$.
	
	This means that moving the potentials $\delta_1,\delta_2,\delta_3$ from $S_1,S_2,S_3$ to $\Delta$ yields $\Delta \geq \frac{1}{4}\arcsin \left( \frac{1}{3} \right) r_m$ while maintaining that $S_1,S_2,S_3$ are saturated, implying that $R_1$ is saturated.
	\end{itemize}
	 This concludes the proof.
\end{proof}

In the proof of Lemma~\ref{lem:technical}, we require that there is a disk packed
into $R_i$ for each $i = 1,\dots,m$, see Lemma~\ref{lem:helping}. In order to
prove this, we need the following technical auxiliary lemma.

\begin{lemma}\label{lem:halfcirclefitsinpocket}
	Let $r_1 \geq 0.495$ placed adjacent to $\mathcal{C}$. The largest disk $r$ whose midpoint can be placed on a tangent $t$ of $r_1$ such that $r_1,r$ do not overlap has a radius of at least $\frac{1}{4}$, see Fig.~\ref{fig:halfdiskfits}.
\end{lemma}
\begin{proof}
	W.l.o.g., we assume $r_1 = 0.495$, because this does not increase  the area between $t$ and $r_1$. Let $m_1$ be the midpoint of $r_1$, $c_1$ the orthogonal projection of $m_1$ onto $t$, and $c_2$ the intersection point of $t$ with the boundary of $\mathcal{C}$, see Fig.~\ref{fig:halfdiskfits}.
	
	\begin{figure}[h!]
  \begin{center}
      \includegraphics[height=4cm]{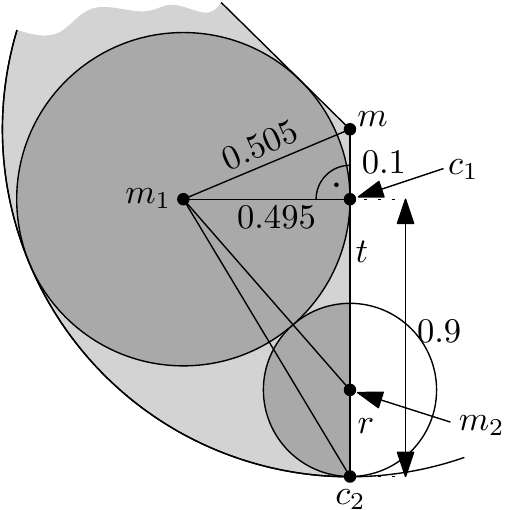} 
  \end{center}
  
  \caption{For $r_1 \geq 0.495$ placed adjacent to $D$, the largest disk $r_2$ with its midpoint on a tangent to $r_1$ has a radius of at least $\frac{1}{4}$.}
  \label{fig:halfdiskfits}
\end{figure}
	
	Pythagoras' Theorem implies $|m c_1| = \sqrt{0.505^2 - 0.495^2} = \frac{1}{10}$, which yields $|c_1 c_2| = \frac{9}{10}$. Assume $\abs{c_2m_2} = \frac{1}{4}$. Another application of Pythagoras' Theorem implies that $|m_1 m_2|$ is equal to $\sqrt{0.495^2 + 0.65^2} \approx 0.81702 > 0.505 + \frac{1}{4} = |m_1 m_2|$. This concludes the proof.
\end{proof}

The pocket below $r_1$ is large enough to ensure that there is at least one disk packed into~$R_{1}$ \new{allowing the following lemma}.

\begin{lemma}\label{lem:helping}
	For each $i = 1,\dots,m$, there is a disk $r_k$ packed into $R_{i}$.
\end{lemma}
\begin{proof}
	W.l.o.g., we assume $i = 1$. Assume there is no disk packed
into~$R_{1}$, i.e., there is a disk $r_k < \frac{1}{4}$ that did not fit
into~$R_{1}$. This implies that all previously packed disks $r_1,\dots,r_{k-1}$
have radii lower bounded by~$\frac{1}{4}$. Thus, $r_1,\dots,r_{k-1}$ are packed
by \new{Boundary} Packing in Phase 1 adjacent to the container disk.
Lemma~\ref{lem:densityCones} implies that the cones induced by
$r_2,\dots,r_{k-1}$ have densities lower bounded by $\rho$, because $r_2 \leq
0.495$. As the cone induced by $r_1$ has an angle of at most $\pi$, the union
$M$ of the cones induced by $r_1,\dots,r_k$ has a density of $\rho$. As $r_k <
\frac{1}{4}$ does not fit between $r_{k-1}$ and $r_1$, the angle between the
upper tangents of $r_{k-1}$ and $r_1$ is smaller than $\arcsin \left(
\frac{r_k}{1-r_k} \right) < \arcsin \left( \frac{1}{3} \right)$. Hence, the
angle of $M$ is larger than $\pi - \arcsin \left( \frac{1}{3} \right)$. Thus,
the potential payed to $\Delta$ is lower bounded by $\pi - \arcsin \left(
\frac{1}{3} \right)$

	 This implies that the entire container disk is saturated, which implies that the total volume of $r_1,\dots,r_n$ is larger than the half of the container disk. This concludes the proof.
\end{proof}

	Next, we consider the case $m >1$, i.e., there are at least two rings with outer radius not smaller than $\frac{1}{4}$. Recall that $r_m$ is the inner radius of $R_m$.

\begin{lemma}\label{lem:removepotentialssmallrings}
	Let $m > 1$. We can remove a potential of $\frac{1}{4} \arcsin \left( \frac{1}{3} \right)r_m^2$ from $R_1,\dots,R_m$ while guaranteeing that $R_1, \dots,R_m$ are saturated. 
\end{lemma}
\begin{proof}
	As the outer radius of $R_m$ is at least as large as $\frac{3}{4}$, the widths of all $R_1,\dots, R_k$ are upper bounded by $\frac{1}{4}$. Thus, the union of all rings $R_1,\dots, R_m$ covers the ring $R[1,\frac{1}{4}]$ and lies inside the ring $R[1,\frac{1}{2}]$. Hence, applying the same approach as used in Lemma~\ref{lem:technical} implies that we can remove a potential of $\frac{1}{4} \arcsin \left( \frac{1}{3} \right)r_m^2$ from $R_1,\dots,R_m$ while guaranteeing that $R_1, \dots,R_m$ is saturated. This concludes the proof.
%\begin{figure}[h!]
%  \begin{center}
%      \includegraphics[width=6cm]{figures/r1large_balance2.pdf} 
%  \end{center}
%  
%  \caption{The case in which the largest maximal ring has a width smaller than $\frac{1}{4}$.}
%  \label{fig:r1large2}
%\end{figure}
\end{proof}

Let $C_\Delta$ be the cone with apex at $m$, angular radius of $\arcsin \left( \frac{1}{3} \right)$, and radius $r_m$ such that $C_\Delta$ touches the upper tangent $t_1$ of $r_1$ from below, see Fig.~\ref{fig:r1large3}(b). 

\begin{figure}[h!]
  \begin{center}
      \includegraphics[width=14cm]{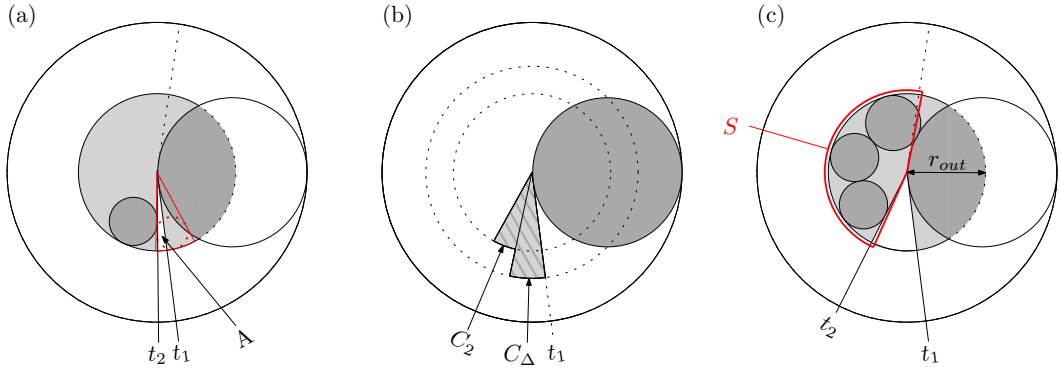} 
  \end{center}
  
  \caption{If there is no disk packed into $R_h$, the complete disk $r_{\text{out}}$ is saturated.}
  \label{fig:r1large3}
\end{figure}

Combining Lemma~\ref{lem:technical} and Lemma~\ref{lem:removepotentialssmallrings} yields the following.

\begin{corollary}\label{cor:removingpotentials}
	We can move a potential of $\frac{1}{4} \arcsin \left( \frac{1}{3} \right)r_m^2$ from~$R_1,\dots,R_m$ to potential function~$\Delta$ while guaranteeing that $R_1,\dots,R_m$ are saturated. 
\end{corollary}

Corollary~\ref{cor:removingpotentials} implies that we can add a potential of
$\frac{1}{4} \arcsin \left( \frac{1}{3} \right) r_m^2$ to $C_\Delta$. This
yields that $C_\Delta$ is saturated, allowing us to show that all the remaining
maximal rings $R_{m+1}, \dots, R_h$ are also saturated.

\begin{lemma}\label{lem:r1middlefinal}
	The rings $R_{m+1}, \dots, R_h$ are saturated.
\end{lemma}
\begin{proof}
	The proof is by induction. Assume that the rings $R_{m+1},\dots,R_{h-1}$ are saturated. Let~$r_{\text{out}}, r_{\text{in}}$ be the outer and inner radius of $R_h$. W.l.o.g., we assume $r_{\text{out}} = 1$.
	
	First we show that there is a disk packed into $R_h$. From this we deduce that $R_h$ is saturated.
	
	\begin{itemize}
		\item For the sake of contradiction, assume there is no disk packed into $R_h$. This means that the disk $r_{\text{crash}}$ that did not fit into $R_h$, is responsible for the construction of $R_h$. Thus, all disks packed previously adjacent to $r_{\text{out}}$ have radii at least as large as $\frac{1}{2}$. Hence, the sector $A$ of $r_{\text{out}}$ between the upper tangent $t_1$ of $r_1$ and the upper tangent $t_2$ of the disk placed last adjacent to $r_{\text{out}}$ (see the white cone in Fig.~\ref{fig:r1large3}(a)) is smaller than the cone induced by $r_{\text{crash}}$, see the red cone in Fig.~\ref{fig:r1large3}(a). Lemma~\ref{lem:densityCones} implies $r_{\text{cash}} < 0.2019 r_{\text{out}}$. Otherwise, the total volume of $r_1,\dots,r_n$ would be larger than $\frac{\pi}{2}$, because the red cone would be saturated in case of $r_{\text{cash}} \geq 0.2019$. This is a contradiction.

	Let $C_2$ be the cone with apex at $m$, an angle of $2\arcsin \left( \frac{0.2019}{0.7981} \right) - \arcsin \left( \frac{1}{3} \right)$, and radius $r_{\text{out}}$ such that $C_2$ touches $C_{\Delta}$ from below, see Fig.~\ref{fig:r1large3}(b). Furthermore, let $S$ be that part of $R_{h}$ that lies above the lower tangent of $r_1$ and below $t_2$, see the red cone in Fig.~\ref{fig:r1large3}(c). As $r_2,r_3,... \leq 0.495$, Lemma~\ref{lem:densityCones} implies that  $S$ has a density of at least $\rho$. As $r_1 \leq \frac{1}{2}$, the cone induced by $r_1$ has an angle of at most~$\pi$. Furthermore, the angle of $A$ is upper bounded by $2\arcsin \left( \frac{0.2019}{0.5 + 0.2019} \right)$, because $r_{\text{crash}} \leq 0.2019$. This implies that $S$ realizes an angle of at least $\pi - 2\arcsin \left( \frac{0.2019}{0.5 + 0.2019} \right)$. Hence, we move a potential of $\left(\rho - \frac{1}{2}\right)\frac{\pi - 2\arcsin \left( \frac{0.2019}{0.5 + 0.2019} \right)}{2 \pi} \pi r_{\text{out}}^2 \geq \frac{0.1549}{2 \pi} \pi r_{\text{out}}^2$ from $S$ to $C_2$, maintaining that $S$ is saturated. Furthermore, $C_2$ is saturated, because $C_2$ has an angle of $2\arcsin \left( \frac{0.2019}{0.5 + 0.2019} \right) \leq 0.2437 < 2 \cdot 0.1549$.
	
	As the angle of $A$ is upper bounded by $2 \arcsin \left(
\frac{0.2019}{0.7019} \right)$, which is smaller than the sum of the angles of
$C_1$ and $C_2$, we obtain that $S$ has a density larger than $\frac{1}{2}$. 
By induction the remaining part of the entire packing container is saturated, so we
conclude that the total input volume is larger than half of the container volume,
which is a contradiction.
	
	\begin{figure}[h!]
  \begin{center}
      \includegraphics[scale=1]{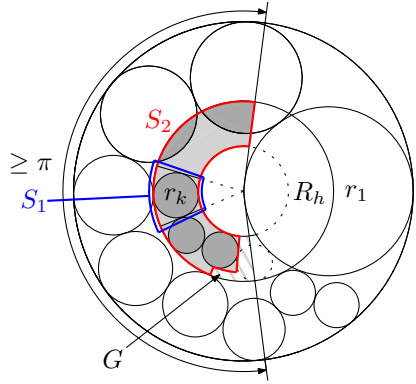} 
  \end{center}
  
  \caption{Guaranteeing that the ring $R_h$ is saturated by upper bounding the volume of the sectors~$E, F$.}
  \label{fig:r1large4}
\end{figure}
	
	\item Lemma~\ref{lem:upperboundlock} implies that the part $G$ of the gap of $R_h$ that lies not inside the cone induced by $r_1$ has a volume of at most $1.07024 U_{R_h}$. In the following, we construct three sectors $S_1,S_2,S_3$ and move certain potentials $\delta_1,\delta_2,\delta_3$ from $S_1,S_2,S_3$ to potential function~$\Delta$. Finally, we show that after moving $\Delta$ to $G$ yields that $G$ is saturated.
	
	\begin{itemize}
		\item Construction of $S_1,\delta_1$: As there is a disk packed into $R_h$, there is a disk $r_k$ touching the inner and the outer boundary of $R_h$, see Fig.~\ref{fig:r1large4}. Let $S_1$ be the sector of $R_h$ between the lower and the upper tangent of $r_k$, see the red sector in Fig.~\ref{fig:r1large4}. Lemma~\ref{lem:zipperlengthone} implies that $S_1$ has a density of at least $0.77036$. Thus, we move a potential of $\left( 0.77036 - \frac{1}{2} \right) \abs{S_1} = 0.27036 U_{R_h}$ from $S_1$ to $\Delta$, which still ensures that $S_1$ is saturated.
	
		\item Construction of $S_2,\delta_2$: Let $S_2$ be the part of $R_h$ that lies between the lower and the upper tangent of $r_1$ and which is covered by sectors of disks already packed by disk or Ring Packing, see the red sector in Fig.~\ref{fig:r1large4}. Corollary~\ref{cor:ringpartrho} implies that $S_2$ as a density of $\rho$. As the angle of the cone induced by $r_1$ is upper bounded by $\pi$, we obtain $\abs{S_2} \geq \frac{\pi}{\pi} \pi \left( 1 - \left( 1 - r_{\text{in}} \right)^2 \right) - 1.07024 U_{R_h}$, which is lower bounded by $3.55196 U_{R_h}$, because $r_{\text{in}} \leq \frac{1}{2}$. Thus, we move a potential of $\delta_2 := \left( \rho - \frac{1}{2} \right)3.55196 U_{R_h} > 0.21542 U_{R_h}$ from $S_2$ to $\Delta$. This yields that $S_2$ is still saturated.
		
		\item Construction of $S_3, \delta$: Let $S_3$ be the intersection of~$C_{\Delta}$ and $R_h$. By construction of~$C_{\Delta}$, see above Corollary~\ref{cor:removingpotentials}, we can move a potential of $\frac{U_{R_h}}{4}$ from $S_3$ to $\Delta$, while ensuring that $R_h \setminus G$ is saturated.
		
		\item Lower-bounding $\Delta$: By the construction of $\delta_1,\delta_2,\delta_3$, we have $\Delta = \delta_1 + \delta_2 + \delta_3$ which is lower-bounded by $\left( 0.27036 + 0.21542 + \frac{1}{4} \right) U_{R_h} = 0.73578 U_{R_h}$, which is enough to saturate a sector of volume $1.47156 U_{R_h}$. As the volume of $G$ is upper bounded by $1.07024 U_{R_h}$, moving $\Delta$ to $G$ saturates $G$ and thus the entire ring $R_h$. This concludes the proof.
	
	\end{itemize}

	\end{itemize}
	
\end{proof}

Combining Lemmas~\ref{lem:removepotentialssmallrings} and~\ref{lem:r1middlefinal} implies that all rings are saturated.

\begin{corollary}\label{cor:ringssaturatedr1middle}
	If $r_1 \geq 0.495$, all rings are saturated
\end{corollary}

Combining Lemma~\ref{lem:largestdisknotadjacenttoboundary} and Corollary~\ref{cor:ringssaturatedr1middle} by the same approach as used in the proof of Lemma~\ref{lem:alldisksfit} implies that all disks are packed concluding the proof of Lemma~\ref{lem:alldisksfit2}.
	
	\restatethm{\lemalldisksfittwo*}{lem:alldisksfit2}

\section{Details of the Analysis for the Case $\frac{1}{2} < r_1$}\label{sec:appendixanalysisr1huge}

In this section we show that our algorithm packs all disks if $\frac{1}{2} < r_1$. In particular, we reduce the case $\frac{1}{2} < r_1$ to the case $r_1 = \frac{1}{2}$. Then an application of the same approach as used for the case $0.495 \leq r_1 \leq \frac{1}{2}$ implies that all disks from the input are packed by our algorithm.

We start with some technical definitions. A \emph{half disk} $H$ is that part of a disk $D$ that lies not to the right of the vertical diameter of $D$. The \emph{midpoint}, the \emph{radius}, and the \emph{vertical diameter} of $H$ are the midpoint, the radius, and the vertical diameter of $D$.

	Let $D$ be the disk with volume $\pi (1 - 1r_1)^2$ inside $\mathcal{C}$ such that $D$ and $r_1$ do not overlap, see Fig.~\new{\ref{fig:reducingr1hugetor1large_main}}. Furthermore, let $H$ by the half disk with radius $1 - 1r_1$ and its diameter crossing orthogonally the touching point between $D$ and $r_1$, see the white half disk in \new{Fig.~\ref{fig:reducingr1hugetor1large_main}}.

\begin{figure}[h!]
  \begin{center}
      \includegraphics[scale=1]{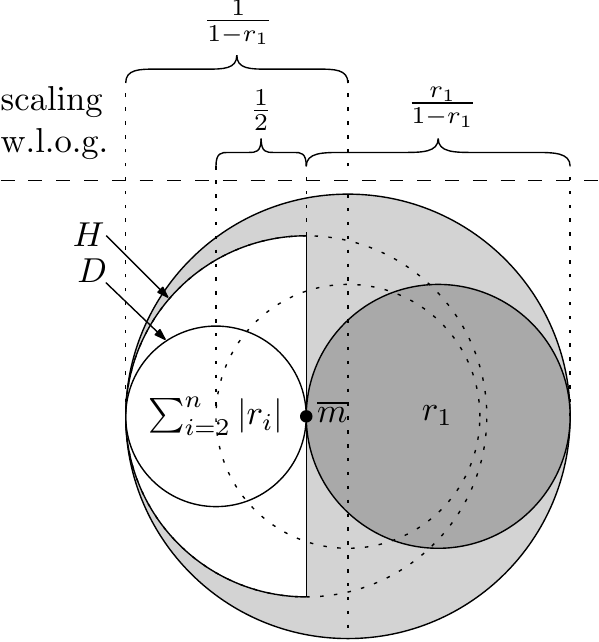} 
  \end{center}
  \caption{The total volume of the remaining disks to be packed is equal to the white disk $D$. As $|H| = 2 |D|$, it suffices to guarantee that $H$ is saturated.}
  \label{fig:reducingr1hugetor1large_main}
\end{figure}

For the remainder of this section, w.l.o.g., we scale our entire configuration such that $D$ has a fixed radius of $\frac{1}{2}$, implying that $\mathcal{C}$ has a radius of $\frac{1}{1-r_1}$ and $r_1$ a new radius of $\frac{r_1}{1-r_1}$, see \new{Fig.~\ref{fig:reducingr1hugetor1large_main}}.

By Lemma~\ref{lem:twodisksonediamater} we are allowed to assume w.l.o.g. that the total volume of $r_2,\dots,r_n$ is $\pi (1 - r_1)^2$.

	The volume of $H$ is \new{at least} twice the volume of $D$, i.e., twice the volume of the remaining disks to be packed. By assumption, \new{$O \setminus \mathcal{C}$} is saturated. As the total input volume is \new{half of} the volume of \new{$O$}, $\mathcal{C} \setminus H$ is saturated.
	
	First, we consider the case that there is no ring created by our algorithm.
	
	\begin{lemma}\label{lem:r1hugenoring}
		If there is no ring created by our algorithm, all input disks are packed by \new{Boundary} Packing.
	\end{lemma}
	\begin{proof} Let $r_i \neq r_1$ be an arbitrary disk, packed by \new{Boundary} Packing and let $m_i$ \new{be} its midpoint. Furthermore, let $\overline{m}$ be the midpoint of $H$, see Fig.~\ref{fig:r1hugenoring}. 
	
\begin{figure}[h!]
  \begin{center}
      \includegraphics[scale=1]{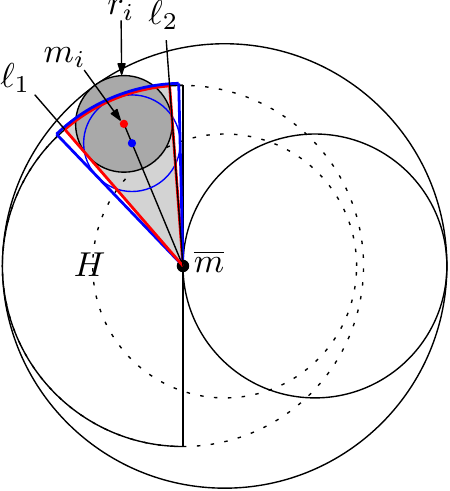} 
  \end{center}
  \caption{In case of $r_1 > \frac{1}{2}$, a disk $r_i$ packed by \new{Boundary} Packing pays its entire potential to that part of $H$ that lies between its two tangents $\ell_1,\ell_2$.}
  \label{fig:r1hugenoring}
\end{figure}

Let $\ell_1,\ell_2$ be the two rays starting in $\overline{m}$ and touching $r_i$. Furthermore, let $S$ be that part of $H$ that lies between $\ell_1,\ell_2$, see the red bounded sector in Fig.~\ref{fig:r1hugenoring}. $r_i$ pays its entire potential to $S$. By moving the midpoint $m_i$ of $r_i$ into the direction of $\overline{m}$ while maintaining the radius of $r_i$, the area of $S$ increases while the volume \new{of $r_i$} stays the same. Lemma~\ref{lem:densityCones} implies that $S$ has a density of at least $\rho$. Finally, applying the same approach as used in the proof of Lemma~\ref{lem:largestdisknotadjacenttoboundary} implies that all disks are packed.
\end{proof}
	
	Thus, we assume w.l.o.g. that there are rings created by our algorithm.

	Let $R_1,\dots,R_{h}$ be the maximal rings ordered decreasingly w.r.t. their outer radii. As $R_1,\dots,R_{h}$ are maximal, $R_1,\dots,R_{h}$ are also ordered decreasingly w.r.t.\ their widths, because our algorithm processes the disks $r_1,\dots,r_n$ in decreasing order. 
	
	\begin{lemma}\label{lem:r1hugesmallestmaximalringlarge}
		If the inner radius of $R_h$ is smaller than $\frac{1}{2}$, all input disks are packed by our approach.
	\end{lemma}
	\begin{proof}
		We use the potential assignments used in the proof of Lemma~\ref{lem:r1hugenoring} and apply the same approach as used in the proof of Lemma~\ref{lem:allmaximalringssmall1}. This concludes the proof.
	\end{proof}
	
\begin{figure}[h!]
  \begin{center}
      \includegraphics[scale=1]{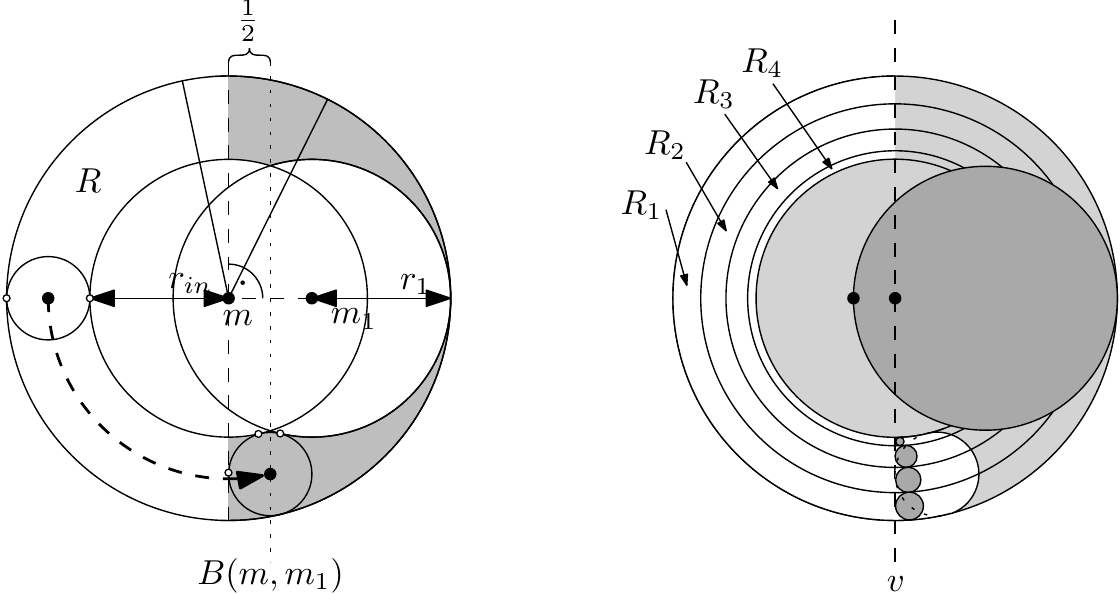} 
  \end{center}
  \caption{If $r_1 = r _{\text{in}}\geq \frac{1}{2}$, the disk with diameter equal to the width of the ring $R$ exactly fits into each of the two gray pockets, i.e., each disk with larger diameter does not fit into a gray pocket.}
  \label{fig:technical_key}
\end{figure}

Hence, we assume w.l.o.g.\ that the inner radius of $R_h$ is smaller than $\frac{1}{2}$. Let $R_1,\dots,R_m$ be all maximal rings such that the outer radius of $R_m$ is not smaller and the inner radius of $R_m$ is smaller than $\frac{3}{4}$. As $R_1,\dots,R_m$ are ordered decreasingly w.r.t.\ their widths, each ring of $R_1,\dots,R_m$ has an inner radius not smaller than $\frac{1}{2}$.

Lemma~\ref{lem:key} implies that each ring $R_1,\dots,R_m$ contains at least one disk \new{because by assumption $r_2, r_3, \dots \leq \frac{1}{2}$}.

\begin{lemma}\label{lem:key}
	Consider a ring $R$ with an inner radius of $r_1$. Let $\ell$ be vertical diameter of $\mathcal{C}$.
	
	The largest disk $\overline{r}$ that can be packed on the same side of $\ell$ as the midpoint of~$r_1$ has a radius of $\frac{1}{2}$, see Fig.~\ref{fig:technical_key}~(Left).
	%two disks $r_1 = r_2 \geq \frac{1}{2}$. Let $r_1$ be placed adjacent to $D$ and $r_2$ with its midpoint in $m$. Furthermore, let $\ell$ be a line passing $m$ perpendicularly to the midpoints of $r_1,r_2$. The largest disks that can be packed overlapping free on the same side of $\ell$ as the midpoint of $r_1$ have radius $\frac{1-r_1}{2}$.
\end{lemma}
\begin{proof}
	Let $m_1$ be the midpoint of $r_1$. The width of $R$ is $\frac{1}{2}$. \new{Note, that after scaling as described above the width remains $\frac{1}{2}$, see Figure~\ref{fig:reducing_main}.} Thus, a disk $\overline{r} := \frac{1}{2}$ inside $R$ touches both the inner and the outer boundary component of $R$. We place the midpoint of $\overline{r}$ on the bisector between $m$ and $m_1$. Thus, $\overline{r}$ is touching the boundary of $r$, because $\overline{r}$ touches the inner boundary component of $R$. Finally, $|mm_1| = 1-r$ implies that the midpoint of $\overline{r}$ has a distance of $\frac{1-r}{2}$ to $\ell$. This concludes the proof.
\end{proof}

\begin{corollary}\label{cor:helping2}
	For each $i = 1,\dots,m$, there is a disk $r_k$ packed into $R_{i}$, see Fig.~\ref{fig:technical_key}~(Right).
\end{corollary}

Note that we assumed w.l.o.g.\ that $\mathcal{C}$ has a radius of $\frac{1}{1-r_1}$. We consider the rings $\overline{R}_1,\dots,\overline{R}_m$ lying inside the container corresponding to the case $r_1 = \frac{1}{2}$ such that $R_i, \overline{R}_i$ have the same width for $i = 1,\dots,m$, \new{see Fig.~\ref{fig:reducing_main}}. Let $L_i,\overline{L}_i$ be that parts of $R_i, \overline{R}_i$ that lie below the horizontal diameters of the container disks and below the lids of $R_i,\overline{R}_i$, see the red sectors in \new{Fig.~\ref{fig:reducing_main}}. 

\begin{figure}[h!]
  \begin{center}
      \includegraphics[scale=1]{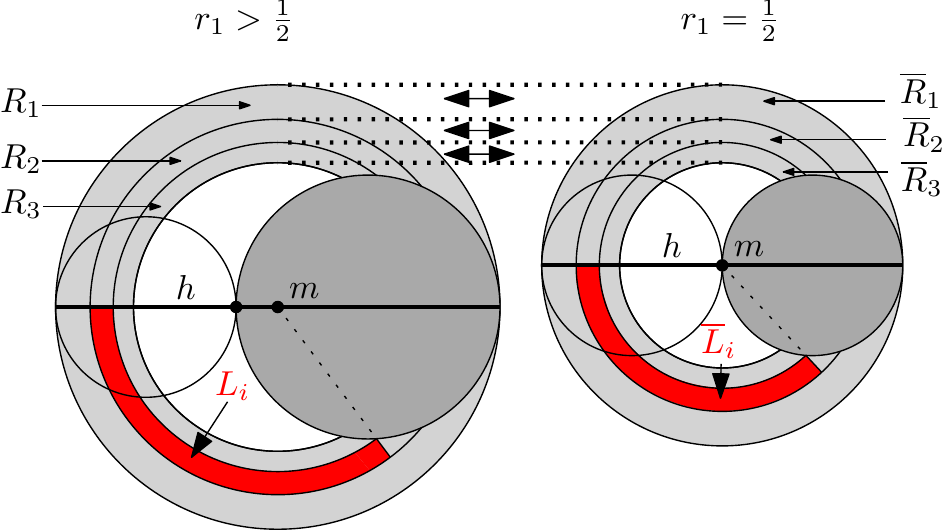} 
  \end{center}
  \caption{Mapping the rings $R_1,\dots,R_m$ onto $\overline{R}_1,\dots \overline{R}_m$. The volume of $\overline{R}_i$ is smaller than the volume of $R_{i}$.}
  \label{fig:reducing_main}
\end{figure}

	In the following, we show that the volume $\abs{\overline{L}_i}$ of $\overline{L}_i$ is not larger than the volume $\abs{L_i}$ of $L_i$. Thus, we need a tool for measuring the volumes of $L_i, \overline{L}_i$. Let $w_i$ be the width of $R_i,\overline{R}_i$. Let $r_{\text{in}}$ be the inner radius if \new{$R_i$}. Furthermore, let $C_i, \overline{C}_i$ be the circles with midpoint $m$ and radius $r_{\text{in}} + \frac{w_i}{2}$ corresponding to the cases $r_1 > \frac{1}{2}$ and $r_1 = \frac{1}{2}$. Analogously, let $c_i \overline{c}_i$ be the circles with midpoint $m$ and radius $r_{\text{in}}$ for $r_1 > \frac{1}{2}$ and $r_1 = \frac{1}{2}$. Furthermore, let $\gamma_i, \overline{\gamma}_i, \mu_i, \overline{\mu}_i$ be those parts of $C_i,\overline{C}_i, c_i, \overline{c}_i$ that lie inside $L_i,\overline{L}_i$, see Fig.~\ref{fig:lengthoflane}. Finally, let $r_1, \overline{r}_1$ the circles presenting the inner circle of $r_1$ for $r_1 > \frac{1}{2}$ and $r_1 = \frac{1}{2}$.

\begin{figure}[h!]
  \begin{center}
      \includegraphics[scale=1]{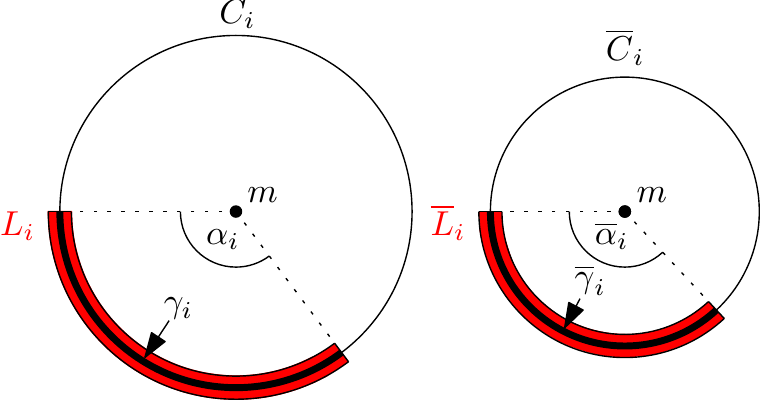} 
  \end{center}
  \caption{The volumes of $L_i$ and $\overline{L}_i$ are $w_i \cdot |\gamma_i|$ and $w_i \cdot |\overline{\gamma}_i|$.}
  \label{fig:lengthoflane}
\end{figure}

\begin{lemma}\label{lem:measuringvolumeoflanes}
	The volumes of $L_i,$ and $\overline{L}_i$ are $w_i \cdot \abs{\gamma_i}$ and $w_i \cdot \abs{\overline{\gamma}_i}$.
\end{lemma}

\begin{lemma}\label{lem:volumerings}
	The volume of $\overline{L}_i$ is not larger than the volume of $L_i$.
\end{lemma}
\begin{proof} As $\overline{L}_i, L_i$ have the same width, Lemma~\ref{lem:measuringvolumeoflanes} implies that we have to show that $\overline{\gamma}_i$ is not longer than $\gamma_i$.

\begin{figure}[h!]
  \begin{center}
      \includegraphics[scale=1]{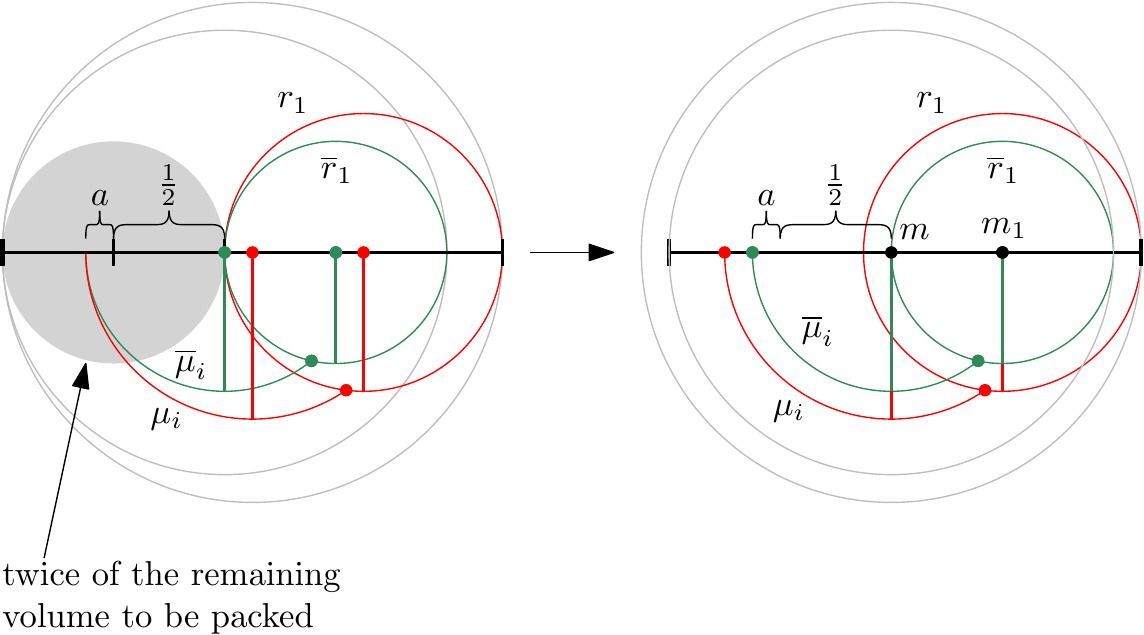} 
  \end{center}
  \caption{We modify the configuration of $\mu_i, \overline{\mu}_i$ and $r_1,\overline{r}_1$ such that $\mu_i, \overline{\mu}_i$ have the same midpoint $m$ and $r_1,\overline{r}_1$ have the same midpoint $m_1$.}
  \label{fig:reducingr1hugetor1large2}
\end{figure}

For simplified calculations, we move $\mathcal{C}, r_1,\mu_i, \gamma_i$ by $r_1 - \frac{1}{2}$ to the left, see the transition illustrated in Fig.~\ref{fig:reducingr1hugetor1large2}. This does not change the lengths of the circular arcs $\gamma_i$, but now all $C_i, c_i, \overline{C}_i, \overline{c}_i, \gamma_i,\overline{\gamma}_i, \mu_i, \overline{\mu}_i$ have the same midpoint $m$ and both versions of $r_1$ for $r_1 > \frac{1}{2}$ and $r_1 = \frac{1}{2}$ have the same midpoint $m_1$, see Fig.~\ref{fig:reducingr1hugetor1large2}~(Right).

\begin{figure}[h!]
  \begin{center}
      \includegraphics[scale=1]{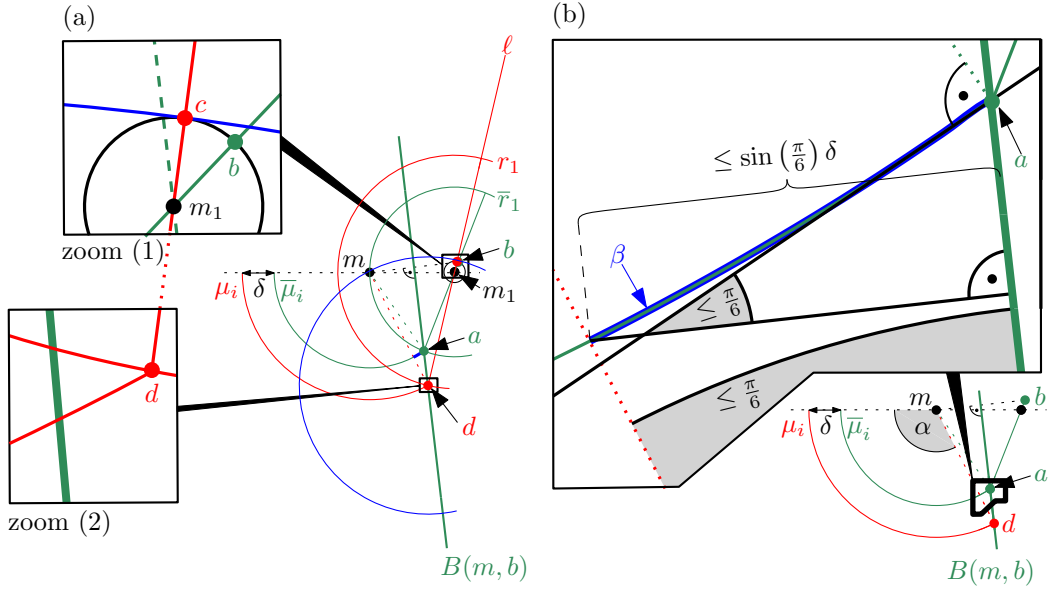} 
  \end{center}
  \caption{(a) In order to upper bound the length of $\gamma_i$ by the length of $\overline{\gamma}_i$, we first show that the end point $d$ of $\mu_i$ lies on the same side of the bisector $B(m,b)$ between $m,b$ as $b$. (b) Then we upper bound the length of the curve $\beta$ by $\frac{\pi/6}{\cos(\pi/6)} \delta$.}
  \label{fig:reducingr1hugetor1large3}
\end{figure}

Let $a$ be the intersection point between $\overline{\mu}_i$ and $\overline{r}_1$, see Fig.~\ref{fig:reducingr1hugetor1large3}(a).  Let $\delta$ be the difference between the radii of $\overline{c}_i$ and $\overline{r}_1$. Furthermore, let $b$ be the point within a distance of $\delta$ to $m_1$ such that $m_1$ lies in the interior of the segment between $a$ and $b$. Let $B(m,b)$ be the bisector between $m$ and $b$. Furthermore, let $d$ be the intersection point between $\mu_i$ and $r_1$ and $\ell$ the line induced by $m_1$ and $d$, see Fig.~\ref{fig:reducingr1hugetor1large3}(a). By construction, $b$ lies closer to $d$ as $c$, see zoom (1) in Fig.~\ref{fig:reducingr1hugetor1large3}(a), implying that $d$ lies on the same side of $B(m,b)$ as $b$, see zoom (2) in Fig.~\ref{fig:reducingr1hugetor1large3}(a). Hence, we assume w.l.o.g. that $\mu_i$ ends on $B(m,b)$.

	The angle induced by $m,a$ in $d$ is upper bounded by $\frac{\pi}{6}$, because the distance between $m,d$ is at least as large as the distance between $m,b$. This implies that the length of that part $\beta$ of $\overline{mu}_i$ that lies on the same side of $B(m,b)$ as $b$, see the blue curve in Fig.~\ref{fig:reducingr1hugetor1large3}(b), is upper bounded by $\frac{\sin \left( \frac{\pi}{6} \right)}{\cos \left( \frac{\pi}{6} \right)} \delta \leq 0.57736 \delta$.
	
	Let $\alpha$ be the left of the two angles between the horizontal diameter of $\mathcal{C}$ and the segment between $m$ and $d$. We have $\alpha \geq \frac{\pi}{2}$, because $b$ lies above the horizontal diameter of $\mathcal{C}$ and to the right of $m$. We also denote by $\mu_i, \overline{\mu}_i$ the radii of $\mu_i, \overline{\mu}_i$. Thus, the length of the circular arc $\mu_i$ is $\frac{\alpha}{2 \pi} \cdot 2 \pi \mu_i = \alpha \overline{\mu}_i + \alpha \delta \geq \alpha \overline{\mu}_i + \frac{\pi}{2} \delta$. This implies that the length of $\gamma_i$ is lower bounded by $\alpha \overline{\mu}_i + \frac{\pi}{2} \delta$
	
	\begin{figure}[h!]
  \begin{center}
      \includegraphics[scale=1]{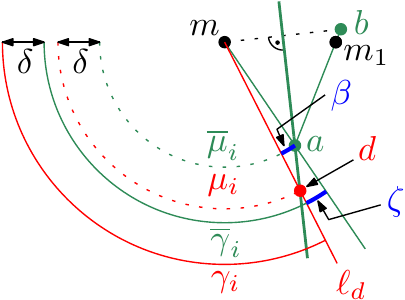} 
  \end{center}
  \caption{The length of the curve $\zeta$ is at most twice the length of the curve $\beta$, because the radius of $\overline{\gamma}_i$ is at most twice the radius of $\overline{\mu}_i$.}
  \label{fig:reducingr1hugetor1large4}
\end{figure}
	
	Let $\ell_d$ be the ray shooting from $m$ into the direction of $d$, see Fig.~\ref{fig:reducingr1hugetor1large4}. This means the lid of $\overline{L}_i$ lies on $\ell_d$. Let $\zeta$ be that part of $\overline{\gamma}_i$ that lies on the same side of $\ell_d$ as $a$. As the radius of $\overline{C}_i$ is at most twice the radius of $\overline{c}_i$, the length of $\zeta$ is at most twice the length of $\beta$. This implies that the length of $\zeta$ is upper bounded by $1.15472 \delta$. 
	
	As $\alpha \geq \frac{\pi}{2}$, $\gamma_i$ is at least as large as $\frac{\pi}{2} \delta$ plus the length of $\overline{\gamma}_i \setminus \zeta$. Thus, we obtain that $\gamma_i$ is not smaller than $\overline{\gamma}_i$, because $\zeta$ is smaller than $\frac{\pi}{2} \delta$.
	
	This concludes the proof.
\end{proof}

Let $C_\Delta$ be the cone inside $H$ with apex at $\overline{m}$, angular radius of $\arcsin \left( \frac{1}{3} \right)$, and radius $r_m$ such that $C_\Delta$ touches the vertical diameter of $H$, see Fig.~\ref{fig:reducingr1hugetor1large}.

\begin{figure}[h!]
  \begin{center}
      \includegraphics[scale=0.8]{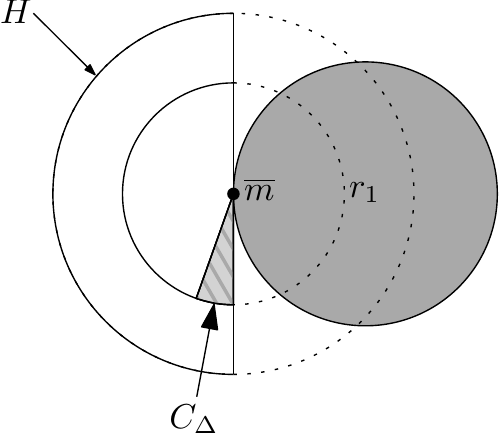} 
  \end{center}
  \caption{The potential that we can remove from $R_1,\dots,R_m$ is large enough to saturate a cone $C_{\Delta}$ inside $H$ with apex at $\overline{m}$, angular radius of $\arcsin \left( \frac{1}{3} \right)$, and radius $r_m$ such that $C_\Delta$ touches the vertical diameter of $H$.}
  \label{fig:reducingr1hugetor1large}
\end{figure}

Combining Corollary~\ref{cor:removingpotentials} and
Lemma~\ref{lem:volumerings} implies that we can move a potential of at least
$\frac{1}{4} \arcsin \left( \frac{1}{3} \right)r_m^2$ from~$R_1,\dots,R_m$
to~$C_{\Delta}$ while guaranteeing that $R_1,\dots,R_m$ are saturated. From
this we deduce that the same approach as used in Lemma~\ref{lem:r1middlefinal}
implies that the maximal remaining rings are also saturated, implying that all
maximal rings are saturated. As each ring is a subset of a maximal ring, we
obtain the following.

\begin{corollary}\label{cor:r1huge_remainingringssaturated}
	All rings are saturated.
\end{corollary}

Let $\ell_m$ be the vertical line touching $r_1$ from the left and $H_m$ that part of the disk container that lies not inside any ring, see Fig.~\ref{fig:reducing}.

\begin{figure}[h!]
  \begin{center}
      \includegraphics[scale=0.8]{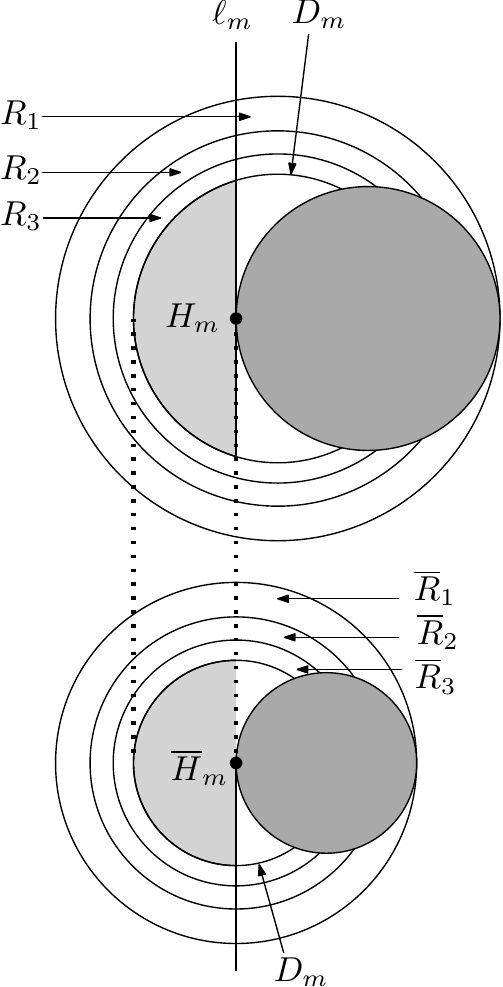} 
  \end{center}
  \caption{The volume of $\overline{H}_m$ is smaller than the volume of $H_m$.}
  \label{fig:reducing}
\end{figure}

Furthermore, let $\overline{H}_m$ be the corresponding sector for $r_1 = \frac{1}{2}$. Let $D_m, \overline{D}_m$ be the disks induced by $H_m, \overline{H}_m$. The half disk $\overline{H}_m$ can be obtained by vertically shrinking $H_m$. Thus, the same approach as used in the proof of Lemma~\ref{lem:alldisksfit} implies that all disks are packed concluding the proof of Lemma~\ref{lem:alldisksfitthree}.

\restatethm{\lemalldisksfitthree*}{lem:alldisksfitthree}

%% file: 06-conclusions.tex
We have established the critical density for packing disks into a disk, based on a number of 
advanced techniques that are more involved than the ones used for packing squares into a square
or disks into a square. Numerous questions remain, in particular the critical density 
for packing disks of bounded size into a disk or the critical density of packing squares into a disk.
These remain for future work; we are optimistic that some of our techniques will be useful.